\documentclass[11pt,reqno]{amsart}

\usepackage{amsmath}
\usepackage[T1]{fontenc}
\usepackage{latexsym}
\usepackage{color}
\usepackage[dvipsnames]{xcolor}
\usepackage[numbers,sort&compress]{natbib}

\usepackage[normalem]{ulem}
\usepackage{cancel}

\usepackage{scalefnt}

\usepackage{textgreek}
\usepackage{mathrsfs}

\usepackage{graphicx}
\usepackage{subcaption}
\usepackage{grffile}

\usepackage{scalerel}

\usepackage{tikz}
\usetikzlibrary{patterns}
\tikzstyle arrowstyle=[scale=1]
\usetikzlibrary{shapes, arrows, fadings,positioning}
\usetikzlibrary{decorations.markings}
\tikzstyle arrowstyle=[scale=1.5]
\tikzstyle directed=[postaction={decorate,decoration={markings,
    mark=at position .65 with {\arrow[arrowstyle]{stealth}}}}]
\tikzstyle reverse directed=[postaction={decorate,decoration={markings,
    mark=at position .65 with {\arrowreversed[arrowstyle]{stealth};}}}]
    \tikzstyle
    dadirected=[postaction={decorate,decoration={markings,
    mark=at position .65 .45  with {\arrow[arrowstyle]{stealth}}}}]
\tikzstyle reverse dadirected=[postaction={decorate,decoration={markings,
    mark=at position .65 .45 with {\arrowreversed[arrowstyle]{stealth};}}}]

\numberwithin{equation}{section}

\newcommand{\R}{\mathbb{R}}
\newcommand{\C}{\mathbb{C}}
\newcommand{\Z}{\mathbb{Z}}

\newcommand{\diff}{\mathrm{d}}
\newcommand{\re}{\mathrm{Re\,}}
\newcommand{\im}{\mathrm{Im\,}}
\newcommand{\ima}{\mathbb{I}}
\newcommand{\tr}{\mathrm{tr}\,}

\newcommand{\ol}{\overline}
\newcommand{\p}{\partial}

\newcommand{\beq}{\begin{equation}}
\newcommand{\eeq}{\end{equation}}

\newcommand{\lp}{\left(}
\newcommand{\rp}{\right)}

\newcommand{\res}{\text{\upshape Res\,}}

\def\Xint#1{\mathchoice
{\XXint\displaystyle\textstyle{#1}}%
{\XXint\textstyle\scriptstyle{#1}}%
{\XXint\scriptstyle\scriptscriptstyle{#1}}%
{\XXint\scriptscriptstyle\scriptscriptstyle{#1}}%
\!\int}
\def\XXint#1#2#3{{\setbox0=\hbox{$#1{#2#3}{\int}$}
\vcenter{\hbox{$#2#3$}}\kern-.5\wd0}}

\def\dashint{\;\Xint-}

\newtheorem{theorem}{Theorem}[section]
\newtheorem{lemma}[theorem]{Lemma}

\newtheorem{remark}[theorem]{Remark}
\newtheorem*{boundary condition}{Boundary condition}
\newtheorem{prop}[theorem]{Proposition}

\title{Exact solution of a Neumann boundary value problem for the stationary axisymmetric Einstein equations}
\author{Jonatan Lenells}

\address{Department of Mathematics, KTH Royal Institute of Technology, 10044 Stockholm, Sweden
}
\email{jlenells@kth.se}

\author{Long Pei}

\address{Department of Mathematics, KTH Royal Institute of Technology, 10044 Stockholm, Sweden
}
\email{longp@kth.se}

\begin{document}

\begin{abstract}
For a stationary and axisymmetric spacetime, the vacuum Einstein field equations reduce to a single nonlinear PDE in two dimensions called the Ernst equation. 
By solving this equation with a {\it Dirichlet} boundary condition imposed along the disk,
Neugebauer and Meinel in the 1990s famously derived an explicit expression for the spacetime metric corresponding to the Bardeen-Wagoner uniformly rotating disk of dust. 
In this paper, we consider a similar boundary value problem for a rotating disk in which a {\it Neumann} boundary condition is imposed along the disk instead of a Dirichlet condition. Using the integrable structure of the Ernst equation, we are able to reduce the problem to a Riemann-Hilbert problem on a genus one Riemann surface. By solving this Riemann-Hilbert problem in terms of theta functions, we obtain an explicit expression for the Ernst potential. Finally, a Riemann surface degeneration argument leads to an expression for the associated spacetime metric. 
\end{abstract}

\maketitle

\setcounter{tocdepth}{1}
\tableofcontents

\section{Introduction}
Half a century ago, Bardeen and Wagoner studied the structure and gravitational field of a uniformly rotating, infinitesimally thin disk of dust in Einstein's theory of relativity \cite{BW1969, BW1971}. Although their study was primarily numerical, they pointed out that there may be some hope of finding an exact expression for the solution. Remarkably, such an exact expression was derived in a series of papers by Neugebauer and Meinel in the 1990s \cite{neugebauer1993einsteinian, neugebauer1994general, neugebauer1995general} (see also \cite{meinel2008relativistic}). Rather than analyzing the Einstein equations directly, Neugebauer and Meinel arrived at their exact solution by studying a boundary value problem (BVP) for the so-called Ernst equation. 

The Ernst equation is a nonlinear integrable partial differential equation in two dimensions which was first written down by F. J. Ernst in the 1960s \cite{E1968}. Ernst made the quite extraordinary discovery that, in the presence of one space-like and one time-like Killing vector, the full system of the vacuum Einstein field equations reduce to a single equation for a complex-valued function $f$ of two variables \cite{E1968}. This single equation, now known as the (elliptic) Ernst equation, has proved instrumental in the study and construction of stationary axisymmetric spacetimes, cf. \cite{klein2005ernst}. 

In terms of the Ernst equation, the uniformly rotating disk problem considered by Bardeen and Wagoner can be reformulated as a BVP for the Ernst potential $f$ in the exterior disk domain $\mathcal{D}$ displayed in Figure \ref{the disk}. Away from the disk, the boundary conditions for this BVP are determined by the requirements that the spacetime should be asymptotically flat, equatorially symmetric, and regular along the rotation axis. On the disk, the requirement that the disk should consist of a uniformly rotating collection of dust particles translates into a {\it Dirichlet} boundary condition for the Ernst potential expressed in a co-rotating frame \cite{meinel2008relativistic}. Neugebauer and Meinel solved this BVP by implementing a series of ingenious steps based on the integrability of the Ernst equation. In the end, these steps led to explicit expressions for both the Ernst potential and the spacetime metric in terms of genus two theta functions. Their solution was ``the first example of solving the boundary value problem for a rotating object in Einstein's theory by analytic methods'' \cite{B2000}.
  
In an effort to understand the Neugebauer-Meinel solution from a more general and systematic point of view, A. S. Fokas and the first author revisited the solution of the above BVP problem in \cite{lenells2010boundary}. It was shown in \cite{lenells2010boundary} that the problem actually is a special case of a so-called {\it linearizable} BVP as defined in the general approach to BVPs for integrable equations known as the unified transform or Fokas method \cite{fokas1997unified}. In this way, the Neugebauer-Meinel solution could be recovered. Later, an extension of the same approach led to the discovery of a new class of explicit solutions which combine the Kerr and Neugebauer-Meinel solutions \cite{lenells2011boundary}. The solutions of \cite{lenells2011boundary} involve a disk rotating uniformly around a central black hole and are given explicitly in terms of theta functions on a Riemann surface of genus four. 

In addition to the BVP for the uniformly rotating dust disk, a few other BVPs were also identified as linearizable in \cite{lenells2010boundary}. One of these problems has the same form as the BVP for the uniformly rotating disk except that a {\it Neumann} condition is imposed along the disk instead of a Dirichlet condition. The purpose of the present paper is to present the solution of this Neumann BVP. Our main result provides an explicit expression for the solution of this problem (both for the solution of the Ernst equation and for the associated spacetime metric) in terms of theta functions on a genus one Riemann surface.
In the limit when the rotation axis is approached, the Riemann surface degenerates to a genus zero surface (the Riemann sphere), which means that we can find particularly simple formulas for the spacetime metric in this limit. 

Our approach can be briefly described as follows. We first use the integrability of the Ernst equation to reduce the solution of the BVP problem to the solution of a matrix Riemann-Hilbert (RH) problem. The formulation of this RH problem involves both the Dirichlet and Neumann boundary values on the disk. By employing the fact that the boundary conditions are linearizable, we are able to eliminate the unknown Dirichlet values. This yields an effective solution of the problem in terms of the solution of a RH problem. However, as in the case of the Neugebauer-Meinel solution, it is possible to go even further and obtain an explicit solution by reducing the matrix RH problem to a scalar RH problem on a Riemann surface. By solving this scalar problem in terms of theta functions, we find exact formulas for the Ernst potential and two of the metric functions. Finally, a Riemann surface condensation argument is used to find an expression for the third and last metric function.

Although our approach follows steps which are similar to those set forth for the Dirichlet problem in \cite{lenells2010boundary} (which were in turn inspired by \cite{neugebauer1993einsteinian, neugebauer1994general, neugebauer1995general}), the Neumann problem considered here is different in a number of ways. 
One difference is that the underlying Riemann surface has genus one instead of genus two. This means that we are able to derive simpler formulas for the spectral functions and for the solution on the rotation axis.
Another difference is that the jump of the scalar RH problem for the Neumann problem does not vanish at the endpoints of the contour. This means that a new type of condensation argument is needed to determine the last metric function. We expect this new argument to be of interest also for other BVPs and for the construction of exact solutions via solution-generating techniques.

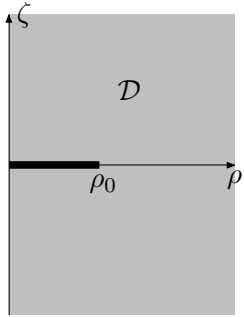
\begin{figure}
\begin{tikzpicture}
\path [fill=lightgray] (0,2) rectangle (3,-2);
\draw [] (0,0) -- (3,0);
\draw[fill] (2.9,0.04) --(2.9,-0.04)--(3,0);
\draw[fill] (0,2) --(-0.04,1.9)--(0.04,1.9);
\draw (0,-2) -- (0,2);
\draw [line width=1mm]  (0,0) -- (1.2,0);
\node at (0.2,2) {$\zeta$};
\node at (1.25,-0.25) {$\rho_{0}$};
\node at (3,-0.2) {$\rho$};
\node at (1.6,1) {$\mathcal{D}$};
\end{tikzpicture}
\caption{\small{The domain $\mathcal{D}$ exterior to a disk of radius $\rho_0 > 0$.}}\label{the disk}
\end{figure}

We do not explore the possible physical relevance of the solved Neumann BVP here. Instead, our solution of this BVP is motivated by the following two reasons: (a) As already mentioned, very few BVPs for rotating objects in general relativity have been solved constructively by analytic methods. 
Our solution enlarges the class of constructively solvable BVPs and expands the mathematical toolbox used to solve such problems. 
(b) An outstanding problem in the context of rotating objects in Einstein's theory consists of finding solutions which describe disk/black hole systems \cite{meinel2008relativistic, klein2005ernst}. The solutions derived in \cite{lenells2011boundary} are of this type. However, the disks in these solutions reach all the way to the event horizon. Physically, there should be a gap between the horizon and the inner rim of the disk (so that the disk actually is a ring). Such a ring/black hole problem can be formulated as a BVP for the Ernst equation with a mixed Neumann/Dirichlet condition imposed along the gap. Thus we expect the solution of a pure Neumann BVP (in addition to the already known solution of the analogous Dirichlet problem) to provide insight which is useful for analyzing ring/black hole BVPs.

\subsection{Organization of the paper}
In Section \ref{prelsec}, we introduce some notation and state the Neumann boundary value problem which is the focus of the paper. The main results are presented in Section \ref{mainresultsec}. In Section \ref{examplesec}, we illustrate our results with a numerical example. 
In Section \ref{laxsec}, we begin the proofs by constructing an eigenfunction $\Phi(z,k)$ of the Lax pair associated with the Ernst equation. We set up a RH problem for $\Phi(z,k)$ with a jump matrix defined in terms of two spectral functions $F(k)$ and $G(k)$. Using the equatorial symmetry and the Neumann boundary condition, we formulate an auxiliary RH problem which is used to determine $F(k)$ and $G(k)$. This provides an effective solution of the problem in terms of the solution of a RH problem. However, as mentioned above, it is possible to obtain a more explicit solution. Thus, in Section \ref{scalarRHsec}, we combine the RH problem for $\Phi$ and the auxiliary RH problem into a scalar RH problem, which can be solved for the Ernst potential $f$. In section \ref{thetasec}, we use tools from algebraic geometry to express $f$ and two of the associated metric functions in terms of theta functions.
In Section \ref{kappasec}, we use a branch cut condensation argument to derive a formula for the last metric function.
In Section \ref{axissec}, we study the behavior of the solution near the rotation axis and complete the proofs of the main results. 

\section{Preliminaries}\label{prelsec}

\subsection{The Ernst equation}
In canonical Weyl coordinates, the exterior gravitational field of a stationary, rotating, axisymmetric body is described by the line element  (cf. \cite{meinel2008relativistic})
\begin{equation}\label{the metric}
\diff s^{2}=e^{-2U}(e^{2\kappa}(\diff\rho^{2}+\diff\zeta^{2})+\rho^{2}\diff\varphi^{2})-e^{2U}(\diff t+a\diff\varphi)^{2},
\end{equation}
where $(\rho,\zeta,\varphi)$ can be thought of as cylindrical coordinates, $t$ as a time variable, and the metric functions $U, \kappa, a$ depend only on $\rho$ and $\zeta$. In these coordinates, the Einstein equations reduce to the system (subscripts denoting partial derivatives)
\begin{align}
&U_{\rho\rho}+U_{\zeta\zeta}+\frac{1}{\rho}U_{\rho}=-\frac{e^{4U}}{2\rho^{2}}(a_{\rho}^{2}+a_{\zeta}^{2}),
	\\
&(\rho^{-1}e^{4U}a_{\rho})_{\rho}+(\rho^{-1}e^{4U}a_{\zeta})_{\zeta}=0\label{Einstein 2}
\end{align}
together with two equations for $\kappa$. In order for the metric \eqref{the metric} to be regular on the rotation axis, the metric functions $a$ and $\kappa$ should vanish on the rotation axis, i.e.,
$$a\to 0,\quad \kappa\to 0\quad \mathrm{as} \quad \rho\to 0.$$
Assuming that the line element \eqref{the metric} approaches the Minkowski metric at infinity (aysmptotic flatness), we also have the conditions
$$U\to 0,\quad a\to 0,\quad  \kappa\to 0 \quad \mathrm{as} \quad \rho^{2}+\zeta^{2}\to \infty.$$
In view of \eqref{Einstein 2}, it is possible to find a function $b(\rho,\zeta)$ which satisfies
\begin{align}\label{b relation 1}
 a_{\rho}=\rho e^{-4U}b_{\zeta},\quad a_{\zeta}=-\rho e^{-4U}b_{\rho}.
\end{align}
The so-called Ernst potential $f := e^{2U}+i b$ then satisfies the Ernst equation
\begin{equation}\label{Ernst equation}
\frac{f+\bar{f}}{2}\left(f_{\rho\rho}+f_{\zeta\zeta}+\frac{1}{\rho}f_{\rho}\right)=f_{\rho}^{2}+f_{\zeta}^{2}
\end{equation}
where $\bar{f}$ denotes the complex conjugate of $f \equiv f(\rho,\zeta)$.

Besides the frame $(\rho,\zeta, \varphi,t)$, we will also use the co-rotating frame $(\rho^{\prime},\zeta^{\prime}, \varphi^{\prime},t^{\prime})$ defined by 
\[
\rho^{\prime}=\rho,\quad \zeta^{\prime}=\zeta,\quad \varphi^{\prime}=\varphi-\Omega t, \quad t^{\prime}=t,
\]
where $\Omega$ denotes the constant angular velocity of the rotating body.  The Ernst equation (\ref{Ernst equation}) and the line element \eqref{the metric} both retain their form in the co-rotating frame. We use the subscript $\Omega$ to denote a quantity in the co-rotating frame; in particular, we let $f_{\Omega} = e^{2U_{\Omega}}+i b_{\Omega}$  denote the Ernst potential in co-rotating coordinates. The co-rotating metric functions $U_{\Omega},\, a_{\Omega},\, \kappa_{\Omega}$ are related to  $U, \, a,\, \kappa$ by (see \cite{meinel2008relativistic})
\begin{subequations}
\begin{align} \label{corotatingmetricfunctionsa}
e^{2U_{\Omega}}&=e^{2U}((1+\Omega a)^{2}-\Omega^{2}\rho^{2}e^{-4U}),
	\\\label{corotatingmetricfunctionsb}
(1-\Omega a_{\Omega})e^{2U_{\Omega}}&=(1+\Omega a)e^{2U},
	\\ 
\kappa_{\Omega}-U_{\Omega}&=\kappa-U.
\end{align}
\end{subequations}

We will use the isomorphism $(\rho,\zeta)\mapsto z := \rho+i \zeta$ to identify $\R^{2}$ and $\C$; in particular, we will often write $f(z) = f(\rho + i\zeta)$ instead of $f(\rho,\zeta)$. In terms of $z$, we have
\begin{align}\label{azkappaz}
a_{z} =i \rho e^{-4U} b_{z} \quad \text{and} \quad
\kappa_{z} =\frac{1}{2}\rho e^{-4U}f_{z}\bar{f}_{z}.
\end{align}

\subsection{The boundary value problem}\label{the boundary value problem section}
Let $\mathcal{D}$ denote the domain exterior to a finite disk of radius $\rho_0 > 0$, that is (see Figure \ref{the disk}),
\[
\mathcal{D}:=\{(\rho,\zeta)\in\R^{2}\,|\, \rho>0\}\setminus \{(\rho,0)\in \R^{2}\,|\, 0<\rho\leq \rho_{0}\}.
\]
In this paper, we consider the following Neumann BVP:
\begin{align}\label{boundary value problem}
\begin{cases}
\text{$f(\rho, \zeta)$ satisfies the Ernst equation \eqref{Ernst equation} in $\mathcal{D}$,} 
	\\ 
\text{$f(\rho,\zeta)=\overline{f(\rho,-\zeta)}$ for $(\rho, \zeta) \in \mathcal{D}$ (equatorial symmetry),}
	\\
\text{$f(\rho,\zeta)\to 1$ as  $\rho^2 + \zeta^2 \to \infty$ (asymptotic flatness),}
	\\
\text{$\frac{\partial f}{\partial \rho}(0, \zeta)=0$ for $\zeta \neq 0$ (regularity on the rotation axis),}
	\\
\text{$\frac{\partial f_{\Omega}}{\partial \zeta}(\rho,\pm 0)=0$ for $0<\rho<\rho_{0}$ (Neumann boundary condition on the disk),}
\end{cases}
\end{align}
where $\rho_0 > 0$ and $\Omega > 0$ are two parameters such that $2\Omega \rho_0 < 1$.

\subsection{The Riemann surface $\Sigma_z$}\label{Riemann surface SigmaZ}
We will present the solution of the BVP (\ref{boundary value problem}) in terms of theta functions associated with a family of Riemann surfaces $\Sigma_z$ parametrized by $z = \rho + i\zeta \in \mathcal{D}$. 
Before stating the main results, we need to define this family of Riemann surfaces. 
In view of the equatorial symmetry, it suffices to determine the solution $f(z)$ of \eqref{boundary value problem} for $z=\rho+i \zeta$ with $\zeta>0$. We therefore assume that $\zeta> 0$ in the following. 

Suppose $\rho_0 > 0$ and $\Omega > 0$ satisfy $2\Omega \rho_0 < 1$. Set $w(k) := -2ik\Omega$ and let 
$$k_1=-\frac{i}{2\Omega}, \qquad \bar{k}_1=\frac{i}{2\Omega}$$ 
denote the two zeros of $w^2-1=0$. 
For each $z\in \C$, we define $\Sigma_{z}$ as the Riemann surface  which consists of all points $(k,y)\in \C^{2}$ such that
\begin{equation}\label{parameterization of Sigmaz}
y^{2} = (k-k_{1})(k-\bar{k}_1)(k+i z)(k-i\bar{z})
\end{equation}
together with two points at infinity which make the surface compact.
We view $\Sigma_{z}$ as a two-sheeted cover of the complex $k$-plane by introducing two branch cuts. The first branch cut runs from $k_1$ to $\bar{k}_1$ and we choose this to be the path $C_{k_1}$ defined by
\begin{align}\label{def cut k1 k1bar}
C_{k_{1}} = [k_1, k_1-1, \bar{k}_1-1,\bar{k}_1]
\end{align}
where $[k_1, k_1-1, \bar{k}_1-1,\bar{k}_1]$ denotes the contour which consists of the straight line segment from $k_1$ to $k_1-1$ followed by the straight line segment from $k_1-1$ to $\bar{k}_1-1$ and so on. The second branch cut runs from $-i z$ to $i\bar{z}$ and we choose this to be the vertical segment $[-i z, i\bar{z}]$. 
The definition (\ref{def cut k1 k1bar}) of $C_{k_1}$ is chosen so that $C_{k_1}$ passes to the left of the vertical contour $\Gamma \subset \C$ defined by 
$$\Gamma = [-i \rho_0, i \rho_0].$$ 
The cut $C_{k_{1}}$ does not intersect $\Gamma$ at the endpoints, because the assumption $2\Omega \rho_0 < 1$ implies that $\rho_0 < |k_1|$. Thus, for each $z \in \mathcal{D}$, the branch cuts and the contour $\Gamma$ are organized as in Figure \ref{uniform procedure near Gamma plus} with $C_{k_{1}}$ and $[-i z, i\bar{z}]$ to the left and right of $\Gamma$, respectively.

We denote by $\Sigma_z^+$ and $\Sigma_z^-$ the upper and lower sheets of $\Sigma_z$, respectively, where the upper (lower) sheet is characterized by $y \sim k^2$ ($y \sim -k^2$) as $k \to \infty$. Let $\hat{\C} = \C \cup \infty$ denote the Riemann sphere.  For $k\in \hat{\C} \setminus (C_{k_1} \cup [-i z, i\bar{z}])$, we write $k^{+}$ and $k^{-}$ for the points in $\Sigma_{z}^{+}$ and $\Sigma_{z}^{-}$, respectively, which project onto $k\in \C$. More generally, we let $A^{+}$ and $A^{-}$ denote the lifts of a subset $A \subset \hat{\C} \setminus (C_{k_1} \cup [-i z, i\bar{z}])$ to $\Sigma_{z}^+$ and $\Sigma_{z}^-$, respectively.

We let $\{a,b\}$ denote the basis for the first homology group $H_{1}(\Sigma_{z},\Z)$ of $\Sigma_{z}$ shown in Figure \ref{homology-3D}, so that $a$ surrounds the cut $C_{k_1}$ while $b$ enters the upper sheet on the right side of the cut $[-i z,i \bar{z} ]$ and exits again on the right side of $C_{k_1}$. 
It is convenient to fix the curves $a$ and $b$ within their respective homology class so that they are invariant under the involution $k^\pm \to k^\mp$. Thus we let $b = [i\bar{z}, k]^+ \cup [k, i\bar{z}]^-$ and let $a$ be the path in the homology class specified by Figure \ref{homology-3D} which as a point set consists of the points of $\Sigma_z$ which project onto $[k_1, \bar{k}_1]$. 
Unless stated otherwise, all contours of integrals on $\Sigma_z$ for which only the endpoints are specified will be supposed to lie within the fundamental polygon obtained by cutting the surface along the curves $\{a,b\}$.
\begin{figure}[ht!]
  \centering
  \includegraphics[width=0.6\linewidth]{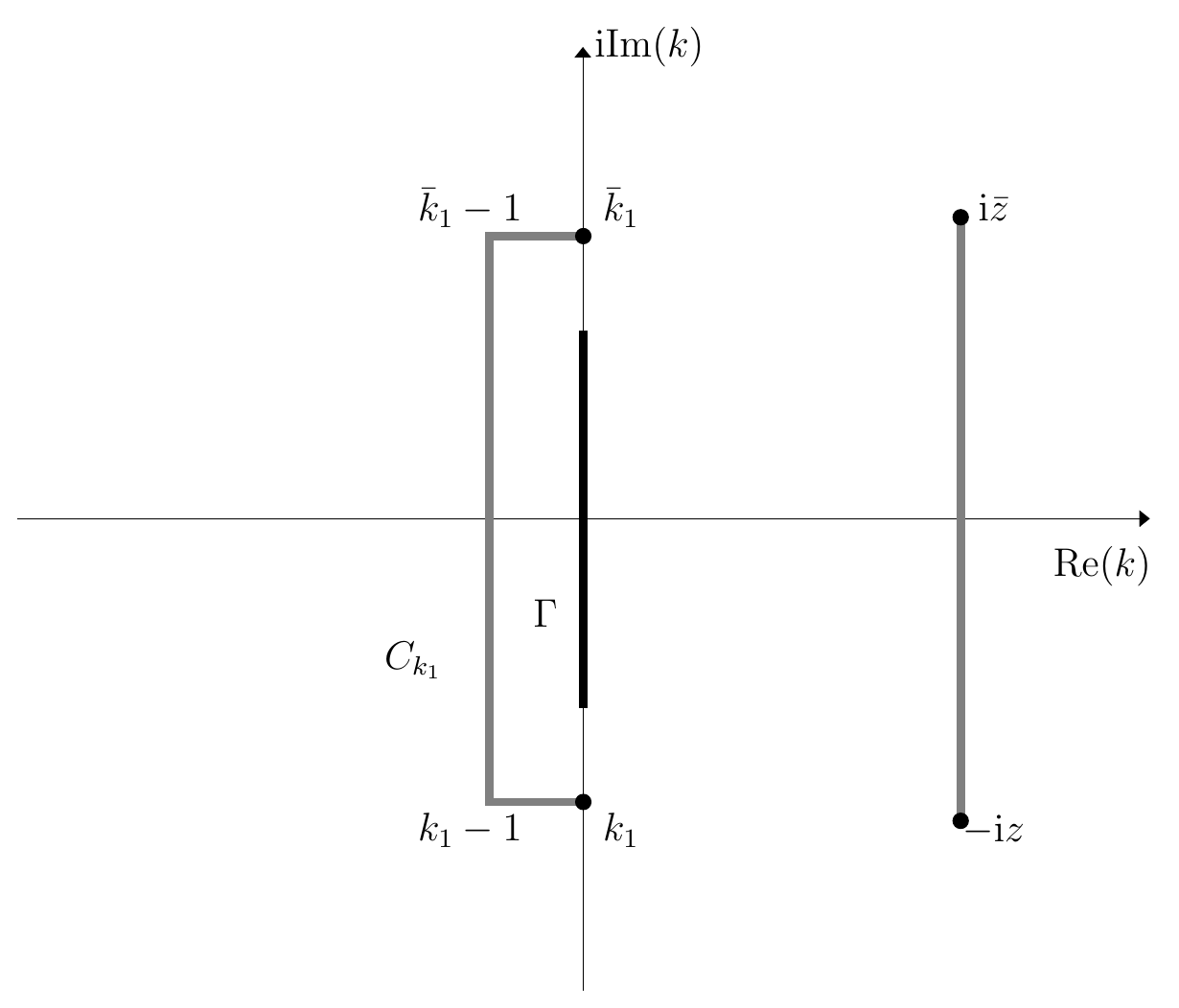}
  \vspace{-.3cm}
\caption{\Small{The contour $\Gamma$  and the branch cuts $C_{k_1}$ and $[-i z,i\bar{z}]$ in the complex $k$-plane.}}
\label{uniform procedure near Gamma plus}
\end{figure}

 \vspace{-.7cm}

\begin{figure}[ht!]
\begin{subfigure}{.45\textwidth}
  \centering
  \includegraphics[width=.7\linewidth]{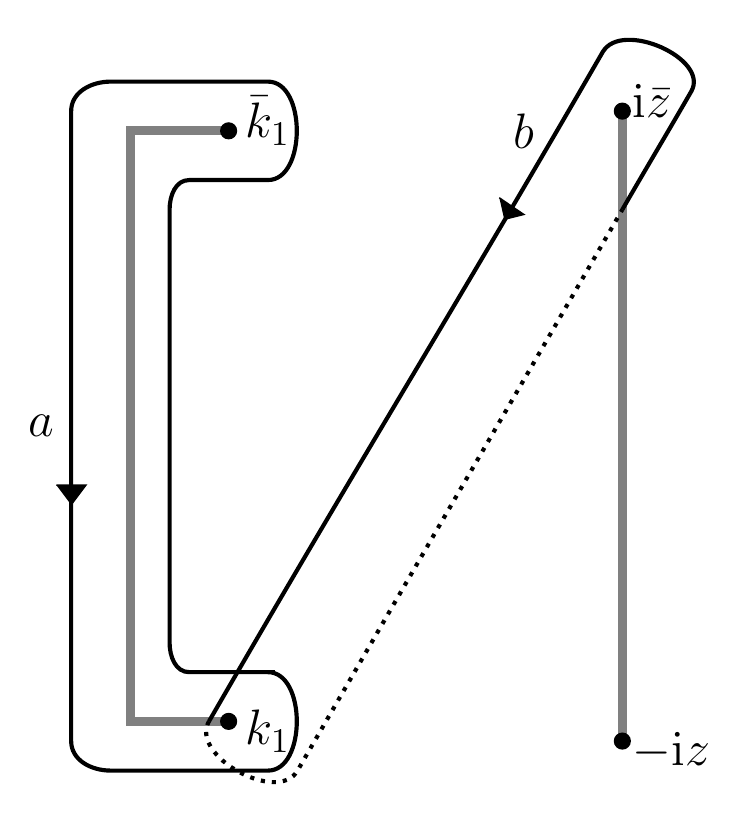}
  \label{homology basis}
\end{subfigure}%
\begin{subfigure}{.5\textwidth}
  \centering
  \includegraphics[width=1.05\linewidth]{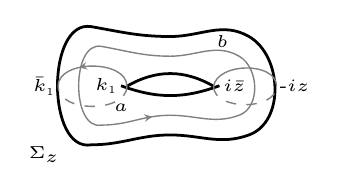}
  \label{3-D view of riemann surface}
\end{subfigure}
\vspace{-.5cm}
\caption{\Small{Two illustrations of the homology basis $\{a,b\}$ on the genus one Riemann surface $\Sigma_z$.}}
\label{homology-3D}
\end{figure}

We let $\omega$ denote the unique holomorphic one-form on $\Sigma_{z}$ such that $\int_{a}\omega=1$. Then the period $B := \int_{b}\omega \in \C$ has strictly positive imaginary part and we may define the Riemann-Siegel theta function $\Theta(v) \equiv \Theta(v|B)$ by
\begin{equation}\label{Thetadef}
\Theta(v|B)=\sum_{N\in \Z}\exp\left[ 2\pi i\left(\frac{1}{2}N^2B+Nv\right)\right], \quad v\in \C.
\end{equation}
Note that 
\begin{equation}\label{omega A eta}
\omega=A\frac{\diff k}{y} \quad \text{and} \quad B=AZ,
\end{equation}
where $A,Z\in \C$ are given by 
\begin{equation}\label{def A}
A^{-1}=\int_{a}\frac{\diff k}{y},\quad Z=\int_{b} \frac{\diff k}{y}.
\end{equation}

We let $\omega_{PQ}$ denote the Abelian differential of the third kind on $\Sigma_{z}$, which has simple poles at the two points $P, Q \in \Sigma_z$ with residues $+1$ and $-1$, respectively, and whose  $a$-period vanishes.  
For $k \in \Sigma_z$, we have
\begin{align}\nonumber
& \omega_{\zeta^{+}\zeta^{-}}(k)=\frac{y(\zeta^{+}) dk}{(k-\zeta)y(k)}-\lp\int_{a}\frac{y(\zeta^{+})dk^\prime}{(k^\prime-\zeta)y(k^\prime)}\rp\omega, 
	\\ \label{Abelian differential third Sigmaz}
& \omega_{\infty^{+}\infty^{-}}(k)=-\frac{k dk}{y(k)}+\left(\int_{a}\frac{k^\prime dk^\prime}{y(k^\prime)}\right)\omega.
\end{align}

\subsection{The Riemann surface $\Sigma'$}\label{degenerated riemann surface}
As $z$ approaches the rotation axis, $\Sigma_z$ degenerates to the $z$-independent Riemann surface $\Sigma^{\prime}$ of genus zero defined by the equation 
\begin{equation}\label{degenerated RS}
\mu = \sqrt{(k-k_1)(k- \bar{k}_1)}.
\end{equation}
We use the branch cut $C_{k_1}$ in (\ref{def cut k1 k1bar}) to view $\Sigma^{\prime}$ as a two-sheeted cover of the complex $k$-plane. We denote the upper and lower sheets of $\Sigma^\prime$ by $\Sigma^{\prime+}$ and $\Sigma^{\prime-}$, respectively, characterized by $y \sim \pm k$ and $y \sim -k$ as $k \to \infty^\pm$.
Letting $\omega^{\prime}_{PQ}$ denote the  Abelian differential of the third kind on $\Sigma^\prime$ with simple poles at $P$ and $Q$, we have the following analog of (\ref{Abelian differential third Sigmaz}):
\begin{equation}\label{Abelian differential third Sigmazprime}
\omega^{\prime}_{\zeta^{+}\zeta^{-}}=\frac{\mu(\zeta^{+}) dk}{(k-\zeta)\mu(k)}, \quad \omega^{\prime}_{\infty^{+}\infty^{-}}=-\frac{dk}{\mu(k)}.
\end{equation}

\section{Main results}\label{mainresultsec}
Define the function $h(k)$ by
\begin{equation}\label{def h}
h(k) = -\frac{\text{arcsin}(2 i \Omega k)}{\pi},
\end{equation}
and note that $h(k)$ is smooth and real-valued for $k \in \Gamma$. 
We also define $u \equiv u(z) \in \C$ and  $ I \equiv I(z) \in \R$ by
\begin{equation}\label{def U and I}
u=\int_{\Gamma^+}h\omega,
\quad I=\int_{\Gamma^+}h\omega_{\infty^{+}\infty^{-}},
\end{equation}
where $\omega$ and $\omega_{\infty^+\infty^-}$ are the differentials on $\Sigma_z$ defined in Section \ref{prelsec} and we view $h$ as a function on $\Gamma^+$ in the natural way, i.e., by composing it with the projection $\Sigma_z \to \C$.

The next theorem, which is our main result, gives an explicit expression for the solution of the boundary value problem (\ref{boundary value problem}) and the associated metric functions in terms of the theta function $\Theta$ associated with the Riemann surface $\Sigma_z$. 

\begin{theorem}[Solution of the Neumann BVP]\label{mainth1}
Suppose  $\Omega>0$ and $\rho_0>0$ are such that $2\Omega\rho_0 < 1$. 
Let $z = \rho + i\zeta$ with $\zeta>0$. Then the solution $f(z) \equiv f(\rho, \zeta)$ of the BVP (\ref{boundary value problem}) is given by
\begin{equation}\label{Ernst potential for f}
f(z)=\frac{\Theta(u-\int_{-i z}^{\infty^{-}}\omega |B)}{\Theta(u+\int_{-i z}^{\infty^{-}}\omega |B)}e^{I},
\end{equation}
where $u$ and $I$ are defined in (\ref{def U and I}).
Moreover, the associated metric functions $e^{2U}$ and $a$ are given by 
\begin{equation}\label{e2U and a}
e^{2U(z)}=\frac{Q(0)}{Q(u)}e^{I},\quad a(z)=a_{0}-\frac{\rho}{Q(0)}\left(\frac{\Theta\left(u+\int_{-i z}^{\infty^{-}}\omega +\int_{i\bar{z}}^{\infty^{-}}\omega | B\right)}{Q(0)\Theta\left(u+\int_{-i z}^{i\bar{z}}\omega | B\right)}-Q(u)\right)e^{-I},
\end{equation}
where $a_0= -\frac{1}{2\Omega}$  and $Q(v)$ is defined by
\begin{equation}\label{def Q}
Q(v) = \frac{\Theta\left( v+\int^{\infty^{-}}_{-i z}\omega|B\right)\Theta\left( v+\int^{\infty^{-}}_{i \bar{z}}\omega|B\right)}{\Theta\left( v|B\right)\Theta\left( v+\int_{-i z}^{i\bar{z}}\omega|B\right)}, \quad v\in \C.
\end{equation}
Finally, the metric function $e^{2\kappa}$ is given by 
\begin{equation}\label{e2kappa}
e^{2\kappa(z)}=K_0\frac{\Theta\left( u|B\right)\Theta\left( u+\int^{i \bar{z}}_{-i z}\omega|B\right)}{\Theta\left( 0|B\right)\Theta\left( \int^{i \bar{z}}_{-i z}\omega|B\right)}e^{L^{reg}},
\end{equation}
where $K_0$  and $L^{reg}$ are given by
\begin{align}\label{K0 formula}
K_0 = &\; \exp\lp-\frac{1}{2}\int_{\Gamma} \diff\kappa_1 h(\kappa_1)\int_{\Gamma}'h(\kappa_2)\lp\frac{\partial\omega^{\prime}_{\kappa^{+}_{1}\kappa^{-}_{1}}}{\partial\kappa_1}(\kappa_2^{+})-\frac{\diff\kappa_2}{(\kappa_2-\kappa_1)^2} \rp\rp,
	\\ \nonumber
L^{reg}= &\; 
\frac{1}{2} \bigg\{\int_{\Gamma} d\kappa_1 h(\kappa_1) \int_{\Gamma}' h(\kappa_2) 
\bigg(\frac{\frac{\diff y}{\diff k}(\kappa_1^+)(\kappa_2-\kappa_1) + y(\kappa_1^+)-y(\kappa_2^+)}{(\kappa_2-\kappa_1)^2y(\kappa_2^+)} \bigg)d\kappa_2
	\\ \label{L def}
& - \bigg(\int_{\Gamma}^{\prime} d\kappa_1 h(\kappa_1) \int_{a} \frac{\frac{\diff y}{\diff k}(\kappa_1^+)(\kappa-\kappa_1) + y(\kappa_1^+)}{(\kappa-\kappa_1)^2 y(\kappa)}d\kappa \bigg)
\bigg(\int_{\Gamma} h(\kappa_2) \omega(\kappa_2^+)\bigg)
\bigg\},
\end{align}
and a prime on an integral along $\Gamma$ means that the integration contour should be slightly deformed before evaluation so that the singularity at $\kappa_1=\kappa_2$ is avoided.\footnote{The result is independent of whether $\Gamma$ is deformed to the left or right of the singularity.}
\end{theorem}

\begin{remark}[Solution for $\zeta \leq 0$]\upshape
Theorem \ref{mainth1} provides expressions for the Ernst potential and the metric functions for $\zeta> 0$. If $\rho > \rho_0$, these expressions extend continuously to $\zeta = 0$. For negative $\zeta$, analogous expressions follow immediately from the equatorial symmetry. In this way, the solution of the BVP (\ref{boundary value problem}) is obtained in all of the exterior disk domain $\mathcal{D}$.
\end{remark}

\begin{remark}[The assumption $2\Omega\rho_0 < 1$]\upshape
We have stated Theorem \ref{mainth1} under the assumption that $2\Omega\rho_0 < 1$. If the rotation speed $\Omega$ and/or the radius $\rho_0$ are so large that $2\Omega\rho_0 \geq 1$, then the branch points $k_1$ and $\bar{k}_1$ lie on $\Gamma$. Nevertheless, the formulas of Theorem \ref{mainth1} can easily be adjusted to include these (possibly singular) solutions.
\end{remark}

\begin{remark}[Solution in terms of elliptic theta functions]\upshape
The Ernst potential $f(z)$ given in (\ref{Ernst potential for f}) can be alternatively expressed in terms of the elliptic theta function $\theta_3$ by
\beq\label{f Jacobi theta functions}
f(z)=\frac{\theta_{3}\big(\frac{\pi (-\int_{-i z}^{\infty^{-}}\omega+u)}{B};e^{- \frac{\pi i}{B}}\big)}{\theta_{3}\big(\frac{\pi (\int_{-i z}^{\infty^{-}}\omega+u)}{B};e^{- \frac{\pi i}{B}}\big)}\exp\bigg( I+\frac{4\pi i \int_{-i z}^{\infty^{-}}\omega}{B}\bigg).
\eeq
The metric functions $e^{2U}$, $a$, and $\kappa$ can also be expressed in terms of elliptic theta functions in a similar way. \end{remark}

\begin{remark}[Behavior at the rim of the disk]\upshape
The Ernst potential $f$ and the metric functions $e^{2U}$, $a$ and $e^{2\kappa}$ given in Theorem \ref{mainth1} are smooth functions of $(\rho, \zeta)$ in the open domain $\mathcal{D}$. They are bounded on $\mathcal{D}$ and extend smoothly to the interior $\{(\rho, 0) \, | \, 0 < \rho < \rho_0\}$ of the disk from both above and below. Moreover, they extend continuously to the rim of the disk (i.e., to the point $(\rho, \zeta) = (\rho_0, 0)$), but they do not, in general, have $C^1$ extensions to this point (cf. Figures \ref{fig:f 3D}-\ref{fig:ae2U and e2kappa}).
In fact, an analysis shows that the boundary values of $e^{2U}$ and of its partial derivatives $\partial_\rho e^{2U}$ and $\partial_\zeta e^{2U}$ on the upper or lower side of the disk extend continuously to the rim of the disk. However, if $\zeta=0$ and $\rho \downarrow \rho_0$, then we only have
\begin{equation*}  
\frac{\p e^{2U}}{\p \rho}(\rho, 0) = O\bigg(\frac{1}{\sqrt{\rho-\rho_0}}\bigg),\qquad \frac{\p e^{2U}}{\p \zeta}(\rho, 0) = O\bigg(\frac{1}{\sqrt{\rho-\rho_0}}\bigg). 
\end{equation*}
Similarly, it can be shown that $a(z)$ and $e^{2\kappa}$ are continuous but not $C^1$ at $(\rho_0, 0)$, and that the partial derivatives $\partial_\rho  a(z)$, $\partial_\zeta a(z)$, $\partial_\rho e^{2\kappa}(z)$, and $\partial_\zeta e^{2\kappa}(z)$ are $O(|z-\rho_0|^{-1/2})$  as $z\to \rho_0$. 
\end{remark}

\subsection{Solution near the rotation axis}
As $z$ approaches the rotation axis (i.e., as $\rho\to 0$), the Riemann surface $\Sigma_z$ degenerates to the genus zero surface $\Sigma'$. We define the quantities $J^{\prime} \equiv J'(\zeta) \in \R$ and $d \equiv d(\zeta) \in i\R_-$ for $\zeta > 0$ by
\begin{align}\label{Jprimedef}
J' = \mu(\zeta^+) \int_{\Gamma} \frac{h(k)}{\mu(k^+)} \frac{dk}{k - \zeta}, \quad
 d = 2i \Omega (\zeta - \mu(\zeta^+)).
\end{align}
where, by (\ref{degenerated RS}),
$$\mu(\zeta^+) = \sqrt{\zeta^2 + \frac{1}{4\Omega^2}} > 0, \quad \zeta > 0,$$
and 
$$\mu(k^+) = \sqrt{-|k|^2 + \frac{1}{4\Omega^2}} > 0, \quad k \in \Gamma.$$

The following result gives the asymptotic behavior of $f$ and the metric functions near the rotation axis.

\begin{theorem}[Solution near the rotation axis]\label{mainth2}
Let $\zeta > 0$. Under the assumptions of Theorem \ref{mainth1}, the following asymptotic formulas hold as $\rho \downarrow 0$:
\begin{itemize}
\item The Ernst potential $f(z)$ satisfies
\begin{equation}\label{fnearaxis}
f(\rho+i\zeta)=f(i \zeta)+O(\rho^{2}),
\end{equation}
where
\begin{equation}\label{axis value of f}
f(i \zeta)=\frac{1 + e^{J^{\prime}}d}{e^{J^{\prime}} + d}. 
\end{equation}

\item The metric functions $e^{2U}$, $a$, and $e^{2\kappa}$ satisfy
\begin{equation}\label{e2U a kappa asymptotic near rotation axis}
e^{2U(\rho+i \zeta)}=e^{2U(i \zeta)}+O(\rho^{2}),\quad
a(\rho+i \zeta)=O(\rho^{2}),\quad
e^{2\kappa(\rho+i \zeta)}=1+O(\rho^{2}),
\end{equation}
where
\begin{equation}\label{e2U axis value}
e^{2U(i \zeta)}=\frac{(1-d^2)e^{J'}}{e^{2J^{\prime}} - d^2}.
\end{equation}
\end{itemize}
\end{theorem}

\begin{remark}[Neumann condition at $\rho=0$]\upshape
Using the results of Theorem \ref{mainth2}, we can easily verify explicitly that the Ernst potential of Theorem \ref{mainth1} satisfies $\re \partial_\zeta f_\Omega(+i0) = 0$, that is, that the real part of the Neumann boundary condition in (\ref{boundary value problem}) holds at the center of the disk.
Indeed, by (\ref{axis value of f}), we have
\begin{align}\label{ddzetarefOmega1}
\frac{\partial}{\partial \zeta}\bigg|_{\zeta = 0^+} f(i\zeta) 
= \frac{\partial}{\partial \zeta}\bigg|_{\zeta = 0^+}\frac{1 + e^{J^{\prime}}d}{e^{J^{\prime}} + d}.
\end{align}
Furthermore, it follows from (\ref{Jprimedef}) that
\begin{align}\label{partialKprime}
\lim_{\zeta \downarrow 0} d = -i, \quad
\lim_{\zeta \downarrow 0} \partial_\zeta d = 2i\Omega,
\end{align}
and
\begin{align*}
\frac{\partial J'}{\partial \zeta}\bigg|_{\zeta = 0^+}  & = 
\lim_{\zeta \downarrow 0}  \int_\Gamma \frac{h(k)}{\mu(k^+)} 
\frac{\partial}{\partial \zeta} \frac{\sqrt{\zeta^2 + \frac{1}{4\Omega^2}}}{k-\zeta} dk
= \lim_{\zeta \downarrow 0} \int_\Gamma \frac{h(k)}{\mu(k^+)}\frac{(1 + 4k \zeta \Omega^2)dk}{4\Omega^2 (k-\zeta)^2 \sqrt{\zeta^2 + \frac{1}{4\Omega^2}}} .
\end{align*}
The  Sokhotski-Plemelj formula gives
\begin{align}\label{partialJprime}
\frac{\partial J'}{\partial \zeta}\bigg|_{\zeta = 0^+}  & = \dashint_\Gamma \frac{h(k)}{\mu(k^+)}\frac{dk}{2\Omega k^2}
- \pi i \underset{k=0}{\res}\bigg( \frac{h(k)}{\mu(k^+)}\frac{dk}{2\Omega k^2} \bigg)
=0 - \frac{1}{\mu(0^+)} = -2\Omega,
\end{align}
where the contribution from the principal value integral vanishes because the function $h(k)$ is odd. 
Using (\ref{partialKprime}) and (\ref{partialJprime}) to compute the right-hand side of (\ref{ddzetarefOmega1}), we find \begin{align}\label{ddzetarefOmega}
\partial_\zeta f(+i0) = 2i\Omega. 
\end{align}
On the other hand, by (\ref{corotatingmetricfunctionsa}) and (\ref{e2U a kappa asymptotic near rotation axis}), we have $\re f_\Omega(i\zeta) =  \re f(i\zeta)$ for $\zeta > 0$. Hence $\re \partial_\zeta f_\Omega(+i0) =  \re \partial_\zeta f(+i0) = \re(2i\Omega) = 0$, which shows that the Neumann boundary condition indeed holds for the real part of $f_\Omega$ at $\rho =0$.
\end{remark}

\section{Numerical example}\label{examplesec}
The formulas of Theorem \ref{mainth1} are convenient for numerical computation. Consider for example the following particular choice of the parameters $\rho_0$ and $\Omega$:
$$\rho_0=1, \qquad \Omega=\frac{3}{10}.$$ 
Then $k_{1}= -5i/3$ and the graphs of the Ernst potential and the metric functions given in Theorem \ref{mainth1} are displayed in Figures \ref{fig:f 3D}-\ref{fig:ae2U and e2kappa}.
It can be numerically verified to high accuracy that $f_{\Omega}$ satisfies the Neumann condition along the disk, and that the defining relations (\ref{azkappaz}) between the metric functions $a, \kappa$, and the Ernst potential $f$ are valid. Similar graphs and numerical results are obtained also for other choices of the parameters $\rho_0$ and $\Omega$ with $2\Omega \rho_0 < 1$.

\begin{figure}[ht!]
\begin{subfigure}{.5\textwidth}
  \centering
  \includegraphics[width=.95\linewidth]{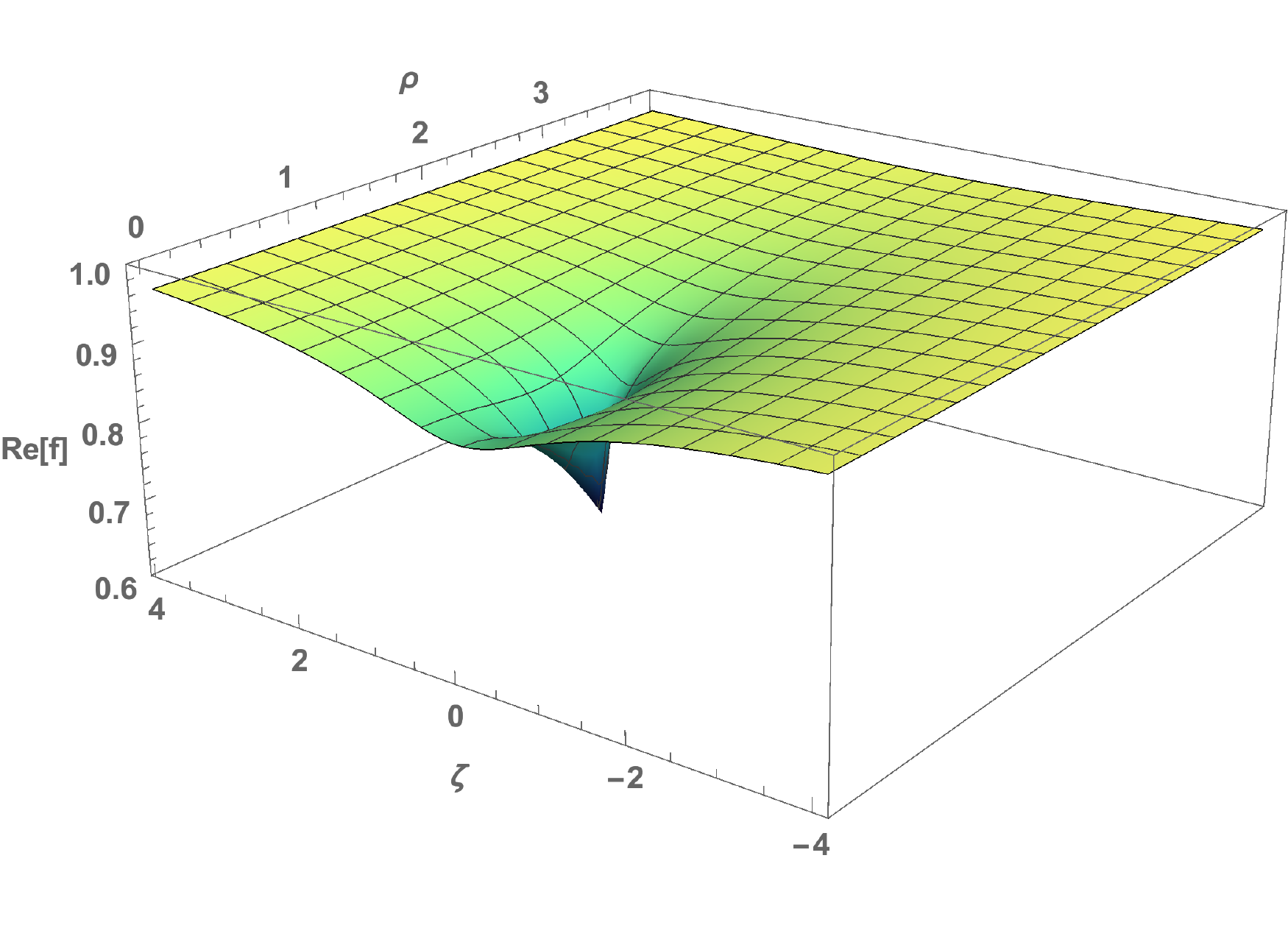}
\end{subfigure}%
\begin{subfigure}{.5\textwidth}
  \centering
  \includegraphics[width=.9\linewidth]{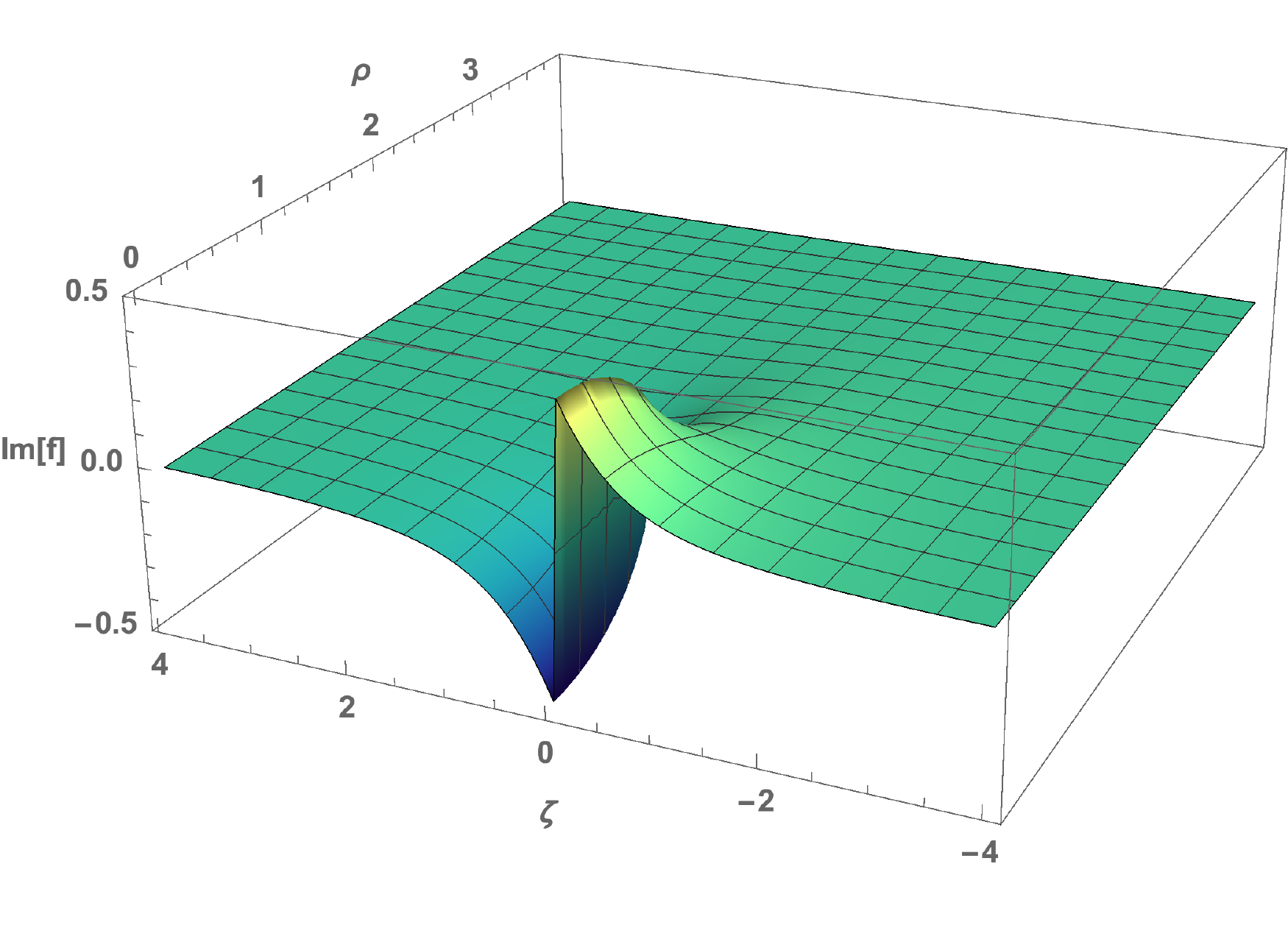}
\end{subfigure}
\caption{\Small{The real and the imaginary parts of the Ernst potential $f$.}}
\label{fig:f 3D}
\end{figure}

\begin{figure}[ht!]
\begin{subfigure}{.5\textwidth}
  \centering
  \includegraphics[width=.9\linewidth]{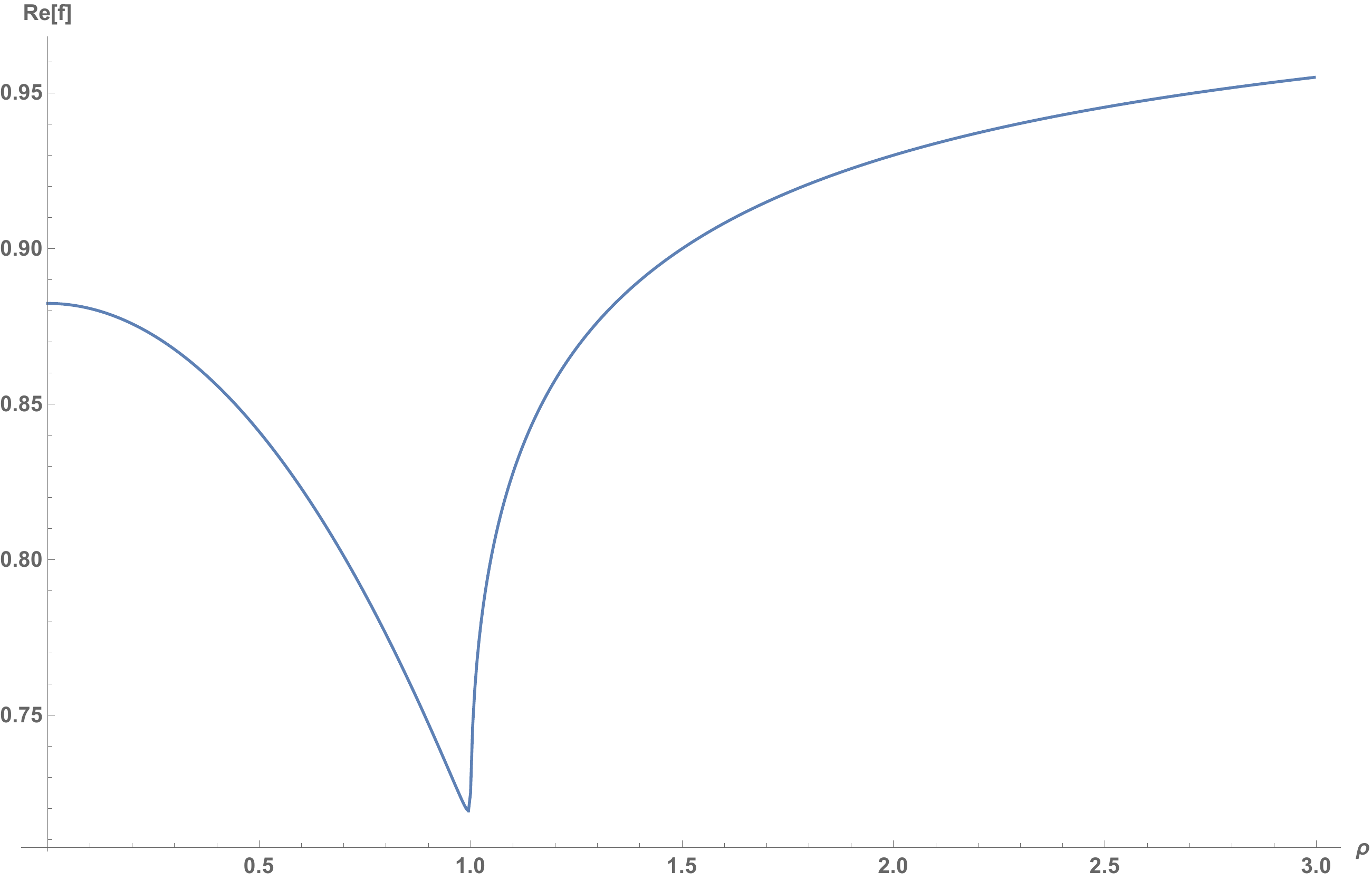}
\end{subfigure}%
\begin{subfigure}{.5\textwidth}
  \centering
  \includegraphics[width=.9\linewidth]{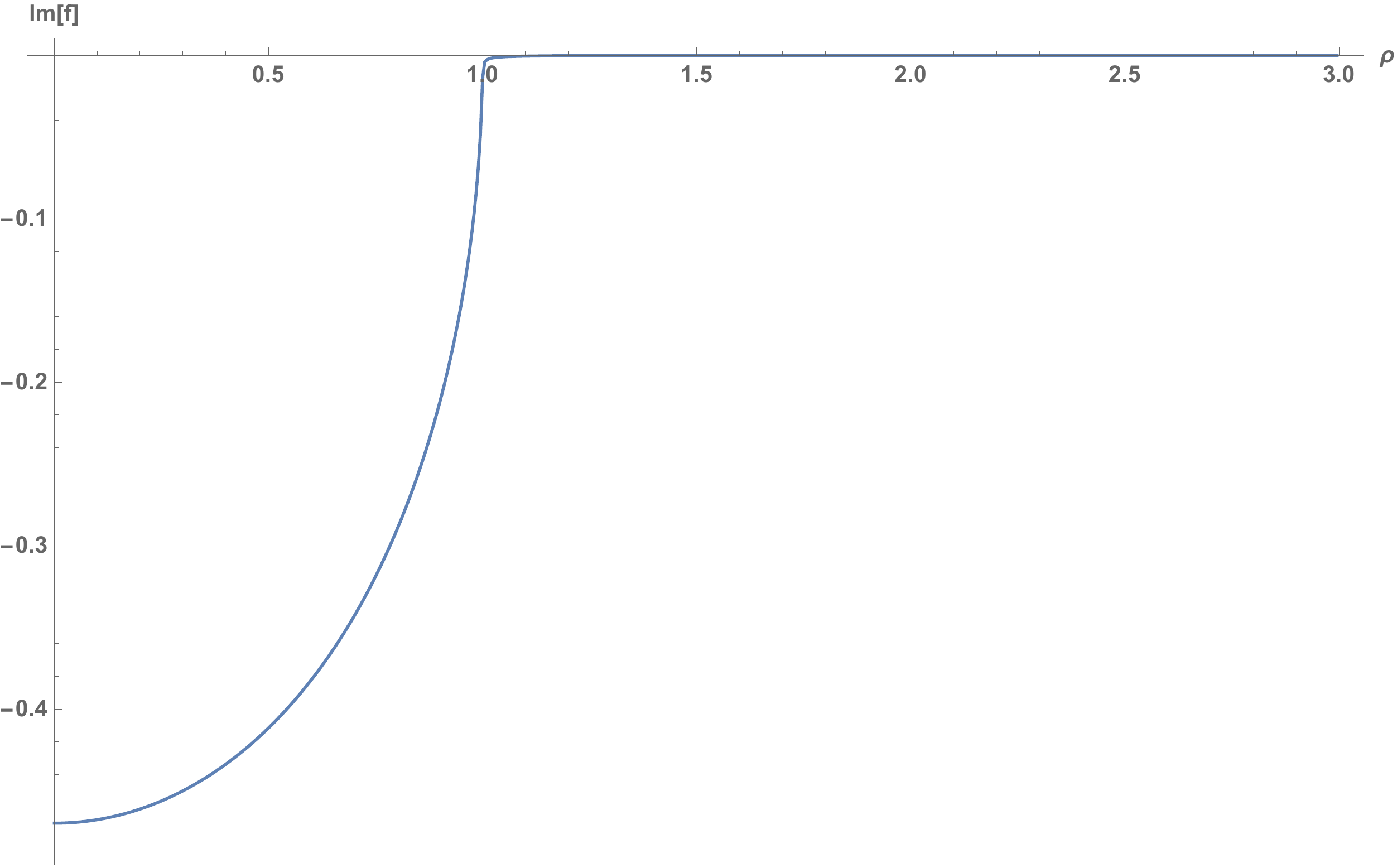}
\end{subfigure}
\caption{\Small{The real and imaginary parts of the Ernst potential $f(\rho + i0)$ on the equatorial plane $\zeta = +0$.}}
\label{fig:f on equatorial plane}
\end{figure}

\begin{figure}[ht]
\vspace{-.4cm}
\begin{subfigure}{.5\textwidth}
  \centering
  \includegraphics[width=.95\linewidth]{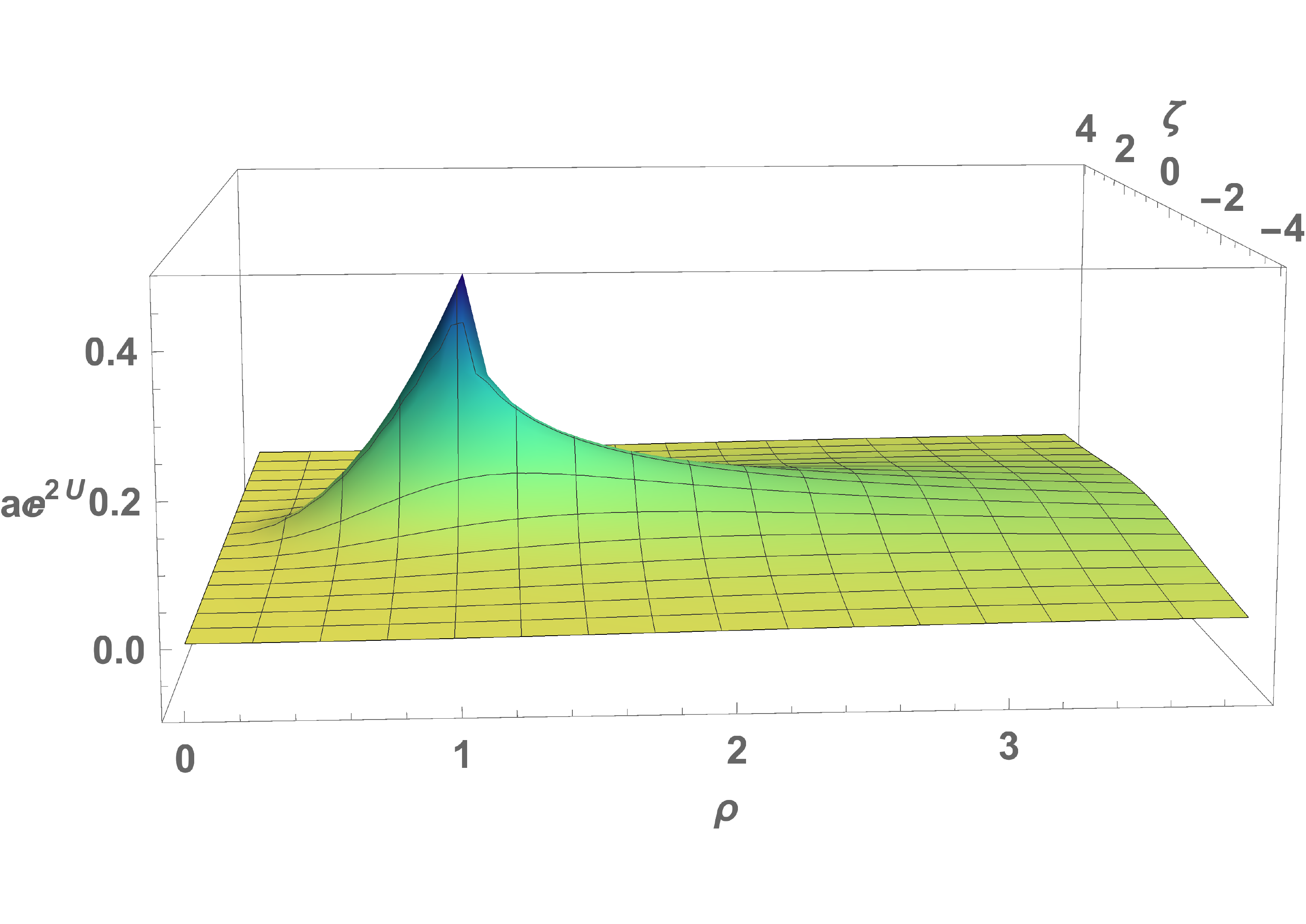}
\end{subfigure}%
\begin{subfigure}{.5\textwidth}
  \centering
  \includegraphics[width=1\linewidth]{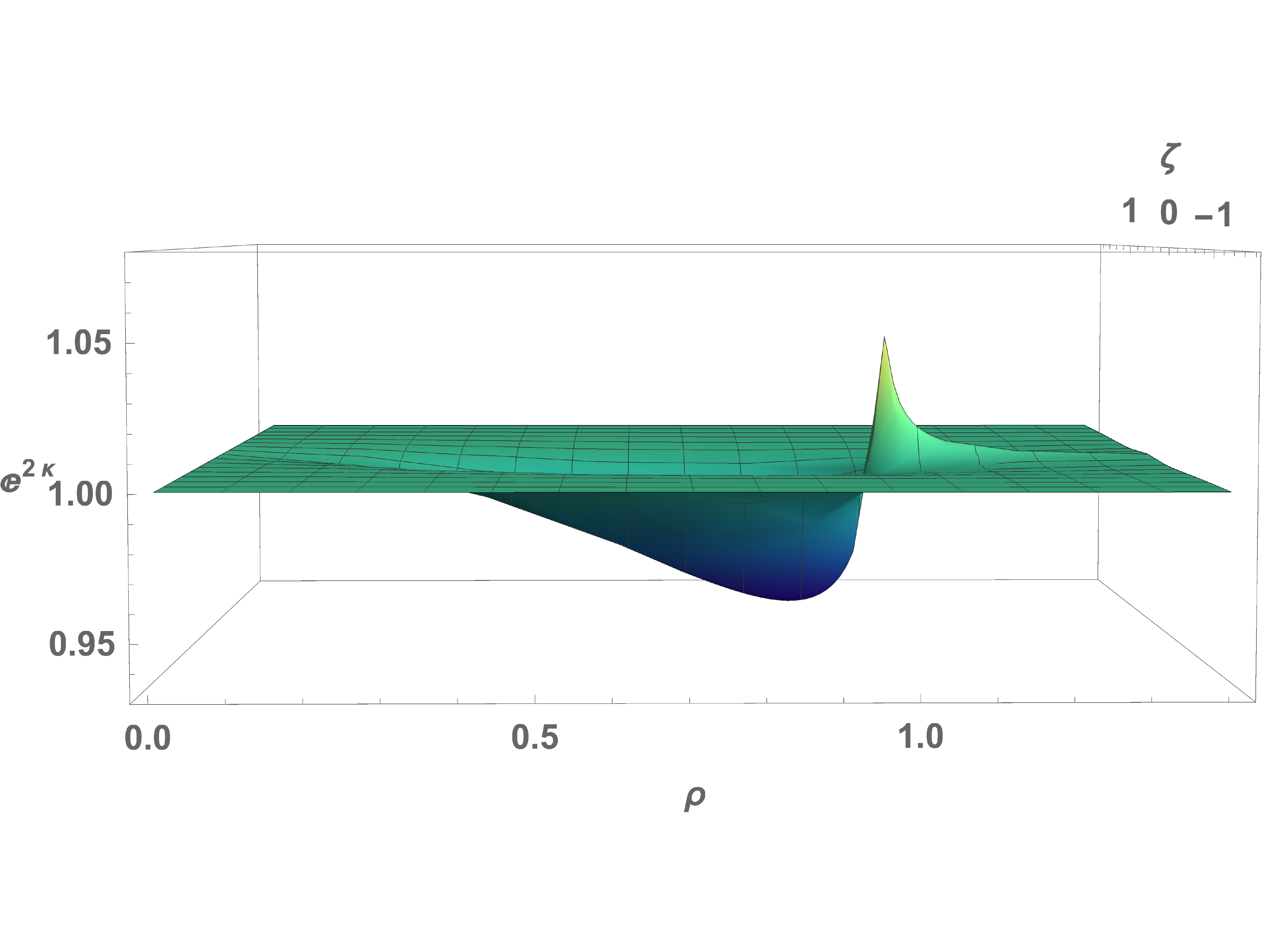}
\end{subfigure}
\vspace{-1cm}
\caption{\Small{The metric functions $ae^{2U}$ and $e^{2\kappa}$.}}
\label{fig:ae2U and e2kappa}
\end{figure}

\section{Lax pair and spectral theory}\label{laxsec}

\subsection{Lax pair}\label{Lax pair formulation}
The stationary Ernst equation \eqref{Ernst equation} admits the Lax pair
\begin{equation}\label{stationary Lax pair}
\begin{split}
\Phi_{z}(z,k) &=U(z,k)\Phi(z,k),\\
\Phi_{\bar{z}}(z,k) &=V(z,k)\Phi(z,k),
\end{split}
\end{equation}
where $z=\rho+i \zeta$,  $\Phi(z,k)$ denotes the $2\times 2$ matrix-valued eigenfuntion, and $U$, $V$ are defined by
\begin{equation}
U =
\frac{1}{f+\bar{f}}\begin{pmatrix}
\bar{f}_{z} & \lambda \bar{f}_{z} \\
\lambda f_{z} & f_{z}  \\
\end{pmatrix},
\quad\quad
V =
\frac{1}{f+\bar{f}}\begin{pmatrix}
\bar{f}_{\bar{z}}& \frac{1}{\lambda}\bar{f}_{\bar{z}} \\
\frac{1}{\lambda}f_{\bar{z}}&f_{\bar{z}}  \\
\end{pmatrix},
\end{equation}
with 
\begin{equation}\label{riemann surface sz}
\lambda = \sqrt{\frac{k-i \bar{z}}{k+i z}}.
\end{equation}
For each $z$, $\Phi(z,\cdot)$ defines a map from $\mathscr{S}_{z}$ to the space of $2\times2$ matrices, where $\mathscr{S}_{z}$ denotes the genus zero Riemann surface defined by (\ref{riemann surface sz}).
We view $\mathscr{S}_{z}$ as a two-sheeted covering of the complex $k$-plane by introducing the branch cut $[-i z, i \bar{z}]$ from $-i z$ to $i \bar{z}$. The upper (lower) sheet of $\mathscr{S}_{z}$ is charactered by $\lambda\to 1$ ($\lambda\to -1$) as $k\to \infty$.
As in the case of $\Sigma_z$ and $\Sigma'$, we write $k^+$ and $k^-$ for the points in $\mathscr{S}_{z}$ which project onto  $k\in \hat{\C} \setminus [-i z, i \bar{z}]$ and which lie on the upper and lower sheets of $\mathscr{S}_{z}$, respectively. 
We let $\Gamma^+$ and $\Gamma^-$ denote the coverings of $\Gamma$ on the upper and lower sheets of $\mathscr{S}_{z}$, respectively. 
For a $2\times2$-matrix $M$, we denote the first and second columns of $M$ by $[M]_{1}$ and $[M]_{2}$, respectively. 

Equation \eqref{stationary Lax pair} can be rewritten in  differential form as
\begin{equation}\label{stationary differential form}
\diff\Phi=W\Phi,
\end{equation}
where the one-form $W$ is  defined by
\begin{equation}\label{stationary W def}
\begin{split}
W&=\frac{1}{f+\bar{f}}\begin{pmatrix}
\bar{f}_{\rho} & \frac{1}{2}\left[\left(\frac{1}{\lambda}+\lambda\right)\bar{f}_{\rho}+i\left(\frac{1}{\lambda}-\lambda\right)\bar{f}_{\zeta}\right]  \vspace{10pt}\\
\frac{1}{2}\left[\left(\frac{1}{\lambda}+\lambda\right)f_{\rho}+i\left(\frac{1}{\lambda}-\lambda\right)f_{\zeta}\right]  & f_{\rho}  \\
\end{pmatrix}\diff\rho \\
&\quad+\frac{1}{f+\bar{f}}\begin{pmatrix}
\bar{f}_{\zeta} & \frac{1}{2}\left[\left(\frac{1}{\lambda}+\lambda\right)\bar{f}_{\zeta}-i\left(\frac{1}{\lambda}-\lambda\right)\bar{f}_{\rho}\right]  \vspace{10pt}\\
\frac{1}{2}\left[\left(\frac{1}{\lambda}+\lambda\right)f_{\zeta}-i\left(\frac{1}{\lambda}-\lambda\right)f_{\rho}\right]  & f_{\zeta}  \\
\end{pmatrix}\diff\zeta\\
&=:W_1 \diff\rho+ W_2\diff \zeta.
\end{split}
\end{equation}
We normalize the eigenfunction $\Phi$ by imposing the conditions 
\begin{equation}\label{stationary initial condition}
\lim_{z\to i \infty}\left[\Phi(z,k^{-})\right]_{1}= \begin{pmatrix}
1 \\
1 
\end{pmatrix},\quad 
\lim_{z\to i \infty}\left[\Phi(z,k^{+})\right]_{2}
=
\begin{pmatrix}
1 \\
-1 
\end{pmatrix}
\end{equation}
for all $k\in \hat{\C}$. As a consequence,  $\Phi(z,k)$  admits the symmetries
\begin{equation}\label{stationary symmetry of Phi from initial condition}
\Phi(z,k^{+})=\sigma_{3}\Phi(z,k^{-})\sigma_{1},\quad 
\Phi(z,k^{+})=\sigma_{1}\overline{\Phi(z,\overline{k}^{+})}\sigma_{3},
\end{equation}
where $\{\sigma_j\}_1^3$ are the standard Pauli matrices. 

We henceforth suppose that $f(z)$ is a solution of the BVP (\ref{boundary value problem}).
Physically, the possibly empty set of points where $\re f = 0$ constitutes the ergospheres of the spacetime. The matrices $U$ and $V$ in the Lax pair \eqref{stationary Lax pair} are, in general, singular at these points and this may give rise to singularities of $\Phi$. On the other hand, points where $\det \Phi = 0$ are related to the presence of solitons (cf. \cite{N1996}). 
It is natural to expect the solution of the boundary value problem (\ref{boundary value problem}) to be free of ergospheres and solitons (at least for small values of $\Omega$). We will therefore henceforth make the assumption that 
\begin{align}\label{ergopheresolitonassumption}
\re f > 0 \quad \text{and} \quad \det \Phi \neq 0
\end{align}
throughout the spacetime. The consistency of this assumption can be ascertained once the final solution has been constructed.
By integrating both sides of \eqref{stationary differential form} from $i \infty$ to $z$ and using (\ref{ergopheresolitonassumption}), it can be confirmed that
$\Phi(z,k)$ is a well-defined analytic function of $k\in \mathscr{S}_{z} \setminus (\Gamma^+\cup \Gamma^-)$ for any fixed $z$, see \cite{lenells2010boundary}.

Let $\Phi_{\Omega}$ denote the co-rotating counterpart of $\Phi$. The Lax pair equations \eqref{stationary Lax pair} and the conditions \eqref{stationary initial condition} retain their form in the co-rotating frame with $f$ and $\Phi$ replaced by $f_{\Omega}$ and $\Phi_{\Omega}$. It is proved in \cite{meinel2008relativistic} that $\Phi_{\Omega}$ is related to $\Phi$ by 
$$\Phi_{\Omega}(z,k)=\Lambda(z,k) \Phi(z,k),$$
where 
$$\Lambda(z,k)=(1+\Omega a)\ima-\Omega\rho e^{-2U}\sigma_3+i (k+i z)\Omega e^{-2U}(-\sigma_3 +\lambda(z,k)\sigma_1\sigma_3),\quad k\in \mathscr{S}_{z},$$
and $\ima$ denote the $2\times2$ identity matrix. 

\subsection{The main RH problem} 
The following lemma can be found in \cite{lenells2010boundary}.
\begin{lemma}\label{Value Phi on zeta axis}
For $k\in \C$,  $\Phi(i\zeta,k)$ can be expressed in terms of $f(i\zeta)$ and the spectral functions $F(k)$ and $G(k)$ as follows:
\begin{align*}
\Phi(i\zeta,k^+) & =\begin{pmatrix}
\ol{f(i\zeta)} & 1 \\
f(i\zeta) & -1  \\
\end{pmatrix}\mathcal{A}(k),\quad \zeta>0,
	\\
\Phi(i\zeta,k^-)&=\begin{pmatrix}
1 & \ol{f(i\zeta)} \\
1& -f(i\zeta)   \\
\end{pmatrix}\sigma_1 \mathcal{A}(k)\sigma_1,\quad \zeta>0,
	\\
\Phi(i\zeta,k^+)&=\begin{pmatrix}
\ol{f(i\zeta)} & 1 \\
f(i\zeta) & -1  \\
\end{pmatrix}\sigma_1 \mathcal{A}(k)\sigma_1,\quad \zeta<0,
	\\
\Phi(i\zeta,k^-)&=\begin{pmatrix}
1 & \ol{f(i\zeta)} \\
1& -f(i\zeta)   \\
\end{pmatrix}\mathcal{A}(k),\quad \zeta<0,
\end{align*}
where $\mathcal{A}(k)$ is defined by 
\begin{equation}\label{def A(k)}
\mathcal{A}(k) = \begin{pmatrix}
F(k) & 
0 \\
G(k) & 1  \\
\end{pmatrix}
\end{equation}
and the spectral functions $F(k)$ and $G(k)$ have the following properties:
\begin{itemize}
\item $F(k)$ and $G(k)$ descend to functions on $\hat{\C}$, namely when viewed as functions on $\mathscr{S}_{z}$, they satisfy
\begin{equation}
F(k^+)=F(k^-),\quad G(k^+)=G(k^-), \quad k\in \hat{\C}.
\end{equation}

\item $F(k)$ and $G(k)$ are analytic for $k\in \C\backslash \Gamma$. 

\item $F(k)=\ol{F(\ol{k})}$ and $G(k)=-\ol{G(\ol{k})}$ for $k\in \C\backslash \Gamma$.

\item As $k\to \infty$,
\begin{equation}\label{asymptotic of F and G as k to infty}
F(k)=1+O(k^{-1}),\quad G(k)=O(k^{-1}).
\end{equation}
\end{itemize}
\end{lemma}
The jumps of the function $\Phi(z,k)$ across $\Gamma^+$ and $\Gamma^-$ can be expressed in terms of $F(k)$ and $G(k)$. In fact, let $g_+(k)$ and $g_-(k)$ denote the boundary values of a function $g(k)$ on the right and left sides of $\Gamma$, respectively. Then $\Phi(z,k)$ satisfies (see \cite{lenells2010boundary,meinel2008relativistic})
\begin{align}\label{Phi jump across Gamma}
\begin{cases}
 \Phi_-(z,k)=\Phi_+(z,k)D(k), & k\in \Gamma^{+},
	\\
 \Phi_-(z,k)=\Phi_+(z,k)\sigma_1 D(k)\sigma_1,& k\in \Gamma^{-},
 \end{cases}
\end{align}
where $D$ is defined by
\begin{align}\label{expression D}
D(k) = \begin{pmatrix}
F_+(k) & 0 \\
G_+(k) & 1  \\
\end{pmatrix}^{-1}
\begin{pmatrix}
F_-(k) & 0 \\
G_-(k) & 1  \\
\end{pmatrix}.
\end{align}

\subsection{The global relation}
Since $f$ is equatorially symmetric, the values of $F(k)$ and $G(k)$ on the left and right sides of $\Gamma$ satisfy an important relation called the global relation.
 
\begin{lemma}[Global relation]
The spectral functions $F(k)$ and $G(k)$ satisfy 
\begin{equation}\label{globalrelation}
\overline{\mathcal{A}_{+}(k)}\sigma_1\overline{\mathcal{A}^{-1}_{+}(k)}\sigma_1=\sigma_1 \mathcal{A}_{-}(k)\sigma_1 \mathcal{A}_{-}^{-1}(k), \qquad k \in \Gamma,
\end{equation}
where $\mathcal{A}_{+}(k)$ and $\mathcal{A}_{-}(k)$ denote the values of $\mathcal{A}(k)$ on the right and left sides of $\Gamma$, respectively.
\end{lemma}
\begin{proof}
See \cite[Proposition 4.3]{lenells2010boundary}.
\end{proof}

\subsection{An additional relation}\label{algebraic relation from Neumann}
The fact that $f_\Omega$ obeys a Neumann condition along the disk implies that $F(k)$ and $G(k)$ satisfy an additional relation beyond the global relation.

We use the superscripts $L$ and $R$ on a function of $k$ to indicate that this function should be evaluated with $k$ lying on the left or right side of the branch cut $[-i z,i\bar{z}]$, respectively. If one of the superscripts $L$ or $R$ is present, we always assume that the evaluation point $k$ lies to the right of $\Gamma$. The latter specification is needed when $z=\rho+i 0$ so that the branch cut $[-i z,i \bar{z}]$ runs infinitesimally close to $\Gamma$.

\begin{lemma}[An additional relation]\label{additionalrelationlemma}
The function $\mathcal{A}(k)\sigma_1 \mathcal{A}^{-1}(k)$ satisfies 
\begin{equation}\label{jump of A sigma1 A inverse}
\begin{split}
\big[\mathcal{B}^{-1}\Lambda^{-1}&(+i 0, k^{+})\sigma_1 \ol{\Lambda(+i 0, k^{+})}\ol{\mathcal{B}}\sigma_1\big]\big[\sigma_1 \overline{\mathcal{A}_{+}(k)}\sigma_1\overline{\mathcal{A}^{-1}_{+}(k)}\sigma_1\big]
	\\
&=-\big[\mathcal{A}_{+}(k)\sigma_1 \mathcal{A}_{+}^{-1}(k)\big]\big[\mathcal{B}^{-1}\Lambda^{-1}(+i 0, k^{+})\sigma_1 \ol{\Lambda(+i 0, k^{+})}\ol{\mathcal{B}}\sigma_1\big],\quad k\in \Gamma,
\end{split}
\end{equation}
where the matrix $\mathcal{B}$ is defined by 
$$\mathcal{B} =
\begin{pmatrix}
\ol{f(+i 0)} & 1 \\
f(+i 0) & -1  
\end{pmatrix}.$$
\end{lemma}
\begin{proof}
We let $W_{\Omega}=W_{1\Omega}\diff\rho+W_{2\Omega}\diff\zeta$ denote the co-rotating analog of the one-form $W$ defined in \eqref{stationary W def}. The restriction of $W_{\Omega}$ to the upper side of the disk is given by $W_{1\Omega}(\rho + i0, k)\diff\rho$.
Evaluating the identity 
\[
\lambda(z,k^{+})=\frac{1}{\lambda(z,(-k+2\zeta)^{+})}
\]
at $z=\rho+i 0$, we find
\begin{equation}\label{left and right lambda}
\lambda^{R}(\rho+i 0,k^{+})=\frac{1}{\lambda^{L}(\rho+i 0,\bar{k}^{+})}, \qquad k \in \Gamma.
\end{equation}
Inserting  the Neumann condition for $f_\Omega$ into the expression (\ref{stationary W def}) for $W_{1\Omega}$ and using \eqref{left and right lambda} as well as the symmetry $\lambda(z,k^\pm) = 1/\overline{\lambda(z,\bar{k}^\pm)}$, we deduce that
$$W_{1\Omega}^R(\rho + i0, k^+) = \sigma_1 \overline{W_{1\Omega}^L(\rho + i0, k^+)} \sigma_1, \qquad \rho \in (0, \rho_0), \ k \in \Gamma.$$
Thus, by \eqref{stationary differential form}, there exists a $2\times 2$-matrix valued, $\rho$-independent function $\mathcal{K}(k)$ such that 
\begin{equation}\label{phi L phi R Q}
\ol{\Phi_{\Omega}^{L}(\rho+i 0, k^{+})}=\sigma_1\Phi_{\Omega}^{R}(\rho+i 0, k^{+})\mathcal{K}(k),\qquad k\in \Gamma.
\end{equation}
It follows from the first symmetry in \eqref{stationary symmetry of Phi from initial condition} that $\sigma_1 \mathcal{K} =-\mathcal{K}\sigma_1$ on $\Gamma$ (cf. \cite[Proposition 5.2]{lenells2010boundary}). The lemma then follows by evaluating \eqref{phi L phi R Q} at  $\rho=0$ and using Lemma \ref{Value Phi on zeta axis}.
\end{proof}

\subsection{The auxiliary RH problem}\label{auxiliary RH problem}
By combining the global relation (\ref{globalrelation}) with the relation of Lemma \ref{additionalrelationlemma}, we can formulate a RH problem for the $2\times 2$-matrix valued function $\mathcal{M}(k)$ by 
\begin{equation}\label{def M A sigma1 A}
\mathcal{M}(k) = \mathcal{A}(k)\sigma_1 \mathcal{A}^{-1}(k)=\begin{pmatrix}
 -G(k)& F(k) \\
\frac{1-G^{2}(k)}{F(k)} & G(k) \\
\end{pmatrix}.
\end{equation}

\begin{lemma}\label{auxiliarylemma}
Suppose $f$ is a solution of the BVP (\ref{boundary value problem}). Then the spectral functions $F(k)$ and $G(k)$ are given by 
\begin{align}\label{FGcalM}
F(k)=\mathcal{M}_{12}(k), \quad G(k)=\mathcal{M}_{22}(k),\quad k\in \C,
\end{align} 
where $\mathcal{M}$ is the unique solution of the following RH problem:
\begin{itemize}
\item $\mathcal{M}(k)$ is analytic for $k\in \C \backslash\Gamma$.
\item Across $\Gamma$, $\mathcal{M}(k)$ satisfies the jump condition
\begin{equation}\label{RH for S and M}
\mathcal{S}(k)\mathcal{M}_-(k)=-\mathcal{M}_+(k)\mathcal{S}(k), \quad k\in \Gamma,
\end{equation}
where $\mathcal{S}(k)$ is defined by 
\begin{equation}\label{expression S}
\mathcal{S}(k) =\begin{pmatrix}
 0& 1 \\
-1 & 4i k\Omega  \\
\end{pmatrix}, \quad k\in \C.
\end{equation}

\item $\mathcal{M}(k)$ has at most logarithmic singularities at the endpoints of $\Gamma$.

\item $\mathcal{M}(k)$ has the asymptotic behavior
\begin{align}
\label{Mm asymptotic as k goes to infty}
\mathcal{M}(k)=\sigma_1+O(k^{-1}), \quad k\to \infty.
\end{align}
\end{itemize}
Moveover,
\begin{equation}\label{trace MS vanish}
\tr{\mathcal{(MS)}=0},\quad k\in \hat{\C}.
\end{equation}
\end{lemma}
\begin{proof}
Since $\det \Phi(z,k) = -2e^{2U(z)}F(k)$ (see \cite[Eq. (2.65)]{meinel2008relativistic}), the assumptions (\ref{ergopheresolitonassumption}) imply that $F(k)$ never vanishes. 
The analyticity of  $\mathcal{M}(k)$  for $k\in \hat{\C}\backslash\Gamma$ then follows from the analyticity of $F(k)$ and $G(k)$ (see Lemma \ref{Value Phi on zeta axis}). 
Combining the global relation (\ref{globalrelation}) with relation \eqref{jump of A sigma1 A inverse}, we find the jump relation \eqref{RH for S and M} with $\mathcal{S}(k)$ given by
\begin{equation}\label{expression of S with k on Gamma}
\mathcal{S}(k)=\mathcal{B}^{-1}\Lambda^{-1}(+i 0,k^+)\sigma_1\ol{\Lambda(+i 0,k^+)}\overline{\mathcal{B}}\sigma_1, \quad k\in \Gamma.
\end{equation}
By inserting the definitions of $\mathcal{B}$ and $\Lambda$, we see that this expression simplifies to  \eqref{expression S} when $k\in \Gamma$.
The asymptotic behavior of $\mathcal{M}$ in \eqref{Mm asymptotic as k goes to infty} is a consequence of the asymptotics of $F$ and $G$ in \eqref{asymptotic of F and G as k to infty}. Define $\mathcal{R}(\rho, k)$ by
\begin{equation}\label{def mathcal R}
\mathcal{R} = \Phi^{-1}_{\Omega}(\rho+i 0, k^{+})\sigma_1 \Phi_{\Omega}(\rho-i 0, k^{+}).
\end{equation}
As in the proof of the existence of the matrix $\mathcal{K}(k)$ in \eqref{phi L phi R Q}, we deduce that $\mathcal{R} $ is independent of $\rho$. Then \eqref{trace MS vanish} follows by evaluating $\mathcal{R}$ at $\rho=\rho_{0}$ and $\rho=0$.
\end{proof}

\subsection{Solution of the auxiliary RH problem}
By solving the RH problem of Lemma \ref{auxiliarylemma}, we can obtain explicit expressions for $F(k)$ and $G(k)$.  

Let $d_1(k)$ denote the unique meromorphic function on $\Sigma'$ which is analytic except for a simple pole at $\infty^-$ with residue $4\Omega i$ and a simple zero at $\infty^+$, i.e.,
$$d_1(k) = 2\Omega i(k - \mu(k))$$
where $\mu(k) = \sqrt{k^2 + 1/(4\Omega^2)}$ is the square root defined in \eqref{degenerated RS}.
Define the scalar-valued function $E(k)$ on $\Sigma'$ by
\begin{align}\label{Edef}
E(k) = e^{\mu(k) \int_{\Gamma} \frac{h(s) }{\mu(s^+)} \frac{ds}{s - k}},
\end{align}
Then $E(k)$ is an analytic function of $k \in \Sigma' \setminus (\Gamma^+ \cup \Gamma^-)$ such that $(i)$ $d_1 E_- + d_1^{-1} E_+  = 0$ for $k \in \Gamma^+ \cup \Gamma^-$, $(ii)$ $E(k)$ has logarithmic singularities at the endpoints of $\Gamma^+$ and $\Gamma^-$, $(iii)$ $E(k^{+}) = \overline{E(\bar{k}^{+})}$, $(iv)$ $E(k^{+}) = E(k^-)^{-1}$, and $(v)$ $E(k^{+}) = 1 + O(k^{-1})$ as $k \to \infty$, where we have used the symmetry $h(k) = -h(-k)$ to see that $E(\infty^+) 
 = 1$.

\begin{lemma}\label{FGexplicitlemma}
The spectral functions $F(k)$ and $G(k)$ are given by
\begin{align}\label{FGexplicit}
\begin{cases}
F(k) 
= \frac{i}{4\Omega \mu(k^+)}\Big(\frac{d_1(k^+)}{E(k^+)} - \frac{E(k^+)}{d_1(k^+)}\Big),
	\vspace{1mm} \\ 
G(k)= \frac{i}{4\Omega\mu(k^+)}\Big( \frac{1}{E(k^+)}- E(k^+)\Big),
\end{cases} \quad
k \in \C \setminus \Gamma.
\end{align}
\end{lemma}
\begin{proof}
The matrix $\mathcal{S}$ in (\ref{RH for S and M}) can be diagonalized as
$\mathcal{S}(k) = T(k) \mathcal{D}(k) T(k)^{-1}$, where
\begin{align*}
& \mathcal{D}(k) = \begin{pmatrix} d_1(k) & 0 \\ 0 & d_2(k) \end{pmatrix},
\qquad T(k) = \begin{pmatrix} d_2(k) & d_1(k) \\ 1 & 1 \end{pmatrix},
\end{align*}
and $d_2(k) := d_1(k)^{-1} = 2\Omega i(k + \mu(k))$. 
We view $\mathcal{D}(k)$ and $\mathcal{T}(k)$ as functions on the Riemann surface $\Sigma^\prime$ and note that $|d_1| = 1$ on $\Gamma^\pm$. The inverse
$$T(k)^{-1} = \frac{1}{4\Omega i \mu(k)}\begin{pmatrix} 1 & -d_1(k) \\ -1 & d_2(k)\end{pmatrix}$$
is analytic on $\Sigma'$ except for simple poles at $k_1, \bar{k}_1$.
It follows from Lemma \ref{auxiliarylemma} that the $2\times 2$-matrix valued function $\mathcal{N}(k)$ defined by
\begin{align}\label{calNdef}
\mathcal{N}(k) = T(k)^{-1} \mathcal{M}(k) T(k)
\end{align}
is an analytic $2\times 2$-matrix valued function of $k \in \Sigma' \setminus (\Gamma^+ \cup \Gamma^- \cup k_1 \cup \bar{k}_1 \cup \infty^+ \cup \infty^-)$ such that $(i)$ $\mathcal{N}(k)$ satisfies the jump condition $\mathcal{D}\mathcal{N}_- + \mathcal{N}_+ \mathcal{D} = 0$ for $k \in \Gamma^+ \cup \Gamma^-$, $(ii)$ $\mathcal{N}(k)$ has at most simple poles at $k_1$ and $\bar{k}_1$, $(iii)$ $\tr \mathcal{N}(k) = 0$ and $\det \mathcal{N}(k) = -1$, and $(iv)$ $\mathcal{N}(k^+) = \sigma_1 \mathcal{N}(k^-) \sigma_1$, 
In terms of $\mathcal{N}$, the identity $\tr(\mathcal{S}\mathcal{M}) = 0$ in (\ref{trace MS vanish}) becomes
$$0 = \tr(\mathcal{S}\mathcal{M}) 
= \tr(\mathcal{D}\mathcal{N}) = (d_1(k) - d_2(k))\mathcal{N}_{11}(k),$$ 
showing that the $(11)$ entry $\mathcal{N}_{11}(k)$ of $\mathcal{N}(k)$ vanishes identically. Since $\tr \mathcal{N}(k) = 0$ and $\det \mathcal{N} = -1$, it follows that $\mathcal{N}_{22}(k) = 0$ and $\mathcal{N}_{12}\mathcal{N}_{21} = 1$.
The $(12)$ entry $\mathcal{N}_{12}$ obeys the jump condition
$$d_1 \mathcal{N}_{12-} + d_2 \mathcal{N}_{12+}  = 0, \qquad k \in \Gamma^+ \cup \Gamma^-,$$
the symmetry $\mathcal{N}_{12}(k^+) = \mathcal{N}_{12}(k^-)^{-1}$, and the asymptotics $\mathcal{N}_{12}(k^-) = 4\Omega i k + O(1)$ as $k \to \infty$.
It follows that the function $n(k)$ defined on $\Sigma'$ by $n(k) = \mathcal{N}_{12}(k)E(k)^{-1}$ has no jump across $\Gamma^+ \cup \Gamma^-$. 
Moreover, since
\begin{align}\label{M22N}
\mathcal{M}_{22} = \frac{\mathcal{N}_{12} + \mathcal{N}_{12}^{-1}}{2} 
+ \frac{k}{2\mu(k)}(\mathcal{N}_{12} - \mathcal{N}_{12}^{-1}),
\end{align}
is analytic at $k_1$ and $\bar{k}_1$, we deduce that $\mathcal{N}_{12}$ (and hence also $n(k)$) actually is analytic at $k_1$ and $\bar{k}_1$. At the endpoints $(\pm i\rho_0)^+$ and $(\pm i\rho_0)^-$ of $\Gamma^\pm$, $n(k)$ has isolated singularities which are at most logarithmic; hence $n(k)$ is analytic also at these points. Consequently, $n(k)$ is the unique meromorphic function on $\Sigma'$ with a simple pole with residue $4\Omega i$ at $\infty^-$ and a simple zero at $\infty^+$, that is, $n(k) = d_1(k)$. We conclude that 
$$\mathcal{N}(k) = \begin{pmatrix} 0 & d_1(k) E(k) \\
d_1(k)^{-1}E(k)^{-1} & 0 \end{pmatrix}.$$
The expressions (\ref{FGexplicit}) follow from (\ref{FGcalM}) and (\ref{calNdef}) by straightforward algebra.
\end{proof}

\subsection{The Ernst potential on the rotation axis}
Recalling that (see \cite[p. 46]{meinel2008relativistic})
\begin{align*}
& F(\zeta) = \frac{1}{\re f(i\zeta)}, \qquad G(\zeta) = \frac{i\im f(i\zeta)}{\re f(i\zeta)}, \qquad \zeta > 0,
	\\
& F(\zeta) = \frac{|f(i\zeta)|^2}{\re f(i\zeta)}, \qquad G(\zeta) = \frac{-i\im f(i\zeta)}{\re f(i\zeta)}, \qquad \zeta <  0,
\end{align*}
we find the following expression for the value of the Ernst potential on the rotation axis:
$$f(i\zeta) = 
\begin{cases}
\frac{1 + G(\zeta)}{F(\zeta)}
= \frac{1 + d_1(\zeta^+) E(\zeta^+)}{d_1(\zeta^+) + E(\zeta^+)}, \qquad \zeta > 0,
	\\
\frac{F(\zeta)}{1 + G(\zeta)}
= \frac{d_1(\zeta^+) + E(\zeta^+)}{1 + d_1(\zeta^+) E(\zeta^+)}, \qquad \zeta < 0.
\end{cases}
$$
Using that $E(\zeta^+) = e^{J'(\zeta)}$ and $d(\zeta) = d_1(\zeta^+)$ for $\zeta > 0$, where $J'$ and $d$ are the functions defined in (\ref{Jprimedef}), we arrive at the expression (\ref{axis value of f}) for $f(i \zeta)$. Since $J'(\zeta) \in \R$ and $d_1(\zeta) \in i\R_-$ for $\zeta > 0$, the expression (\ref{e2U axis value}) for $e^{2U(i \zeta)}$ follows by taking the real part of (\ref{axis value of f}).

\section{The scalar RH problem}\label{scalarRHsec}
Substitution of the expressions for $F(k)$ and $G(k)$ obtained in Lemma \ref{FGexplicitlemma} into (\ref{expression D}) gives an explicit expression for the jump matrix $D(k)$. Thus we have an effective solution of the BVP (\ref{boundary value problem}) in terms of the solution of the matrix RH problem (\ref{Phi jump across Gamma}). 
In what follows, we instead employ the main and auxiliary RH problems to formulate a scalar RH problem on the Riemann surface $\Sigma_z$. The solution of this scalar RH problem leads to the exact formulas of Theorem \ref{mainth1}.

\subsection{The functions $\mathcal{L}(z,k)$ and $\mathcal{Q}(z,k)$}\label{L and Q} 
Define $w(k)$ by
$$w(k) = -\frac{1}{2}\tr(\mathcal{S}(k)) = -2ik\Omega,\quad k\in \C.$$
Introduce the $2\times2$ matrix valued functions $\mathcal{L}$ and $\mathcal{Q}$ by
\begin{align}
\mathcal{L}(z,k)&=\Phi(z,k)\sigma_1\Phi^{-1}(z,k), \quad k\in \mathscr{S}_{z},
	\\ \label{calQdef}
\mathcal{Q}(z,k)&=-\Phi(z,k)\mathcal{A}^{-1}(k)\mathcal{S}(k) \mathcal{A}(k)\Phi^{-1}(z,k)-w(k)\ima, \quad k\in \mathscr{S}_{z}.
\end{align}

\begin{lemma}\label{properties of L and Q}
The functions $\mathcal{L}$ and $\mathcal{Q}$ have the following properties:
\begin{itemize}
\item $\mathcal{L}$ and $\mathcal{Q}$ satisfy the trace and determinant relations
\begin{equation}\label{trace determinant of L and Q}
\tr\mathcal{Q} =0, \quad \tr \mathcal{L}=0, \quad \det \mathcal{L} =-1,\quad \det\mathcal{Q}=1-w^{2},
\end{equation}
and the symmetries
\begin{equation}\label{qm km kp lm km kp}
\mathcal{Q}(z,k^-)=-\sigma_3\mathcal{Q}(z,k^+)\sigma_3,\quad \mathcal{L}(z,k^-)=\sigma_3\mathcal{L}(z,k^+)\sigma_3.
\end{equation}

\item $\mathcal{Q}$ can be rewritten as
\begin{equation}\label{another form of qm}
\mathcal{Q}(z,k)=\Phi(z,k)\sigma_1 \mathcal{A}^{-1}(k)\mathcal{S}(k) \mathcal{A}(k)\sigma_1\Phi^{-1}(z,k)+w(k)\ima.
\end{equation}

\item $\mathcal{Q}$ has no jump across $\Gamma^+ \cup \Gamma^-$, whereas $\mathcal{L}$ satisfies the jump conditions 
\begin{align}\label{QLjumps}
\begin{cases}
(\mathcal{Q}+w\ima)\mathcal{L}_{-} =-\mathcal{L}_{+}(\mathcal{Q}+w\ima), & k\in \Gamma^+,
	\\
(\mathcal{Q}-w\ima)\mathcal{L}_{-} =-\mathcal{L}_{+}(\mathcal{Q} -w\ima), & k\in \Gamma^-.
\end{cases}
\end{align}

\item $\mathcal{L}$ and $\mathcal{Q}$ anticommute, i.e., $\mathcal{Q}\mathcal{L}=-\mathcal{L}\mathcal{Q}$.

\item The function $\hat{\mathcal{L}}$ defined by
$$\hat{\mathcal{L}} =\mathcal{L}\begin{pmatrix}
1 & \mathcal{Q}_{11} \\
0 & \mathcal{Q}_{21}  \\
\end{pmatrix},$$
satisfies
\begin{equation}\label{Lhat22 L21 w Q21 multiply}
\hat{\mathcal{L}}_{22}^{2}-\mathcal{L}_{21}^{2}(w^{2}-1)=\mathcal{Q}^{2}_{21}.
\end{equation}

\item For each $z$, there exists a point $m_{1} \equiv m_{1}(z) \in \C$ such that 
\begin{equation}\label{q21 formula}
\mathcal{Q}_{21}(z,k)=\frac{4i \Omega f}{f+\bar{f}}(k-m_{1}).
\end{equation}
\end{itemize}
\end{lemma}

\begin{proof}
The properties in \eqref{trace determinant of L and Q} follow from the definitions of $\mathcal{L}$ and $\mathcal{Q}$ and  the fact that $\det(\mathcal{S} + w\ima) = 1 - w^2$. The symmetries in \eqref{qm km kp lm km kp} are a consequence of \eqref{stationary symmetry of Phi from initial condition} and (\ref{trace MS vanish}).
The alternative expression \eqref{another form of qm} for $\mathcal{Q}$ follows from (\ref{trace MS vanish}).
Equations \eqref{Phi jump across Gamma} and \eqref{expression D} imply that $\Phi \mathcal{A}^{-1}$ and $\Phi\sigma_1 \mathcal{A}^{-1}$ do not jump across $\Gamma^{+}$ and $\Gamma^{-}$, respectively. Thus, by (\ref{calQdef}) and (\ref{another form of qm}), $\mathcal{Q}$ does not jump across $\Gamma^+$ or $\Gamma^-$. Since, for $k \in \Gamma^+$,
\begin{align}
(\mathcal{Q}_{-}+w\mathbb{I})\mathcal{L}_{-}&=-\Phi_{-}\mathcal{A}^{-1}_{-}\mathcal{S}\mathcal{A}_- \sigma_1\Phi_{-}^{-1}=-\Phi_{-}\mathcal{A}^{-1}_{-}\mathcal{S}\mathcal{M}_{-} \mathcal{A}_{-}\Phi^{-1}_{-},
	\\
-\mathcal{L}_{+}(\mathcal{Q}_{+}+w\mathbb{I})&=\Phi_{+}\sigma_1 \mathcal{A}^{-1}_{+}\mathcal{S}\mathcal{A}_{+}\Phi_{+}^{-1}=\Phi_{+} \mathcal{A}^{-1}_{+}\mathcal{M}_{+}\mathcal{S}\mathcal{A}_{+}\Phi_{+}^{-1},
\end{align}
the jump of $\mathcal{L}$ across $\Gamma^{+}$ is a consequence of (\ref{Phi jump across Gamma}) and \eqref{RH for S and M}.
The  jump  of $\mathcal{L}$ across  $\Gamma^{-}$ then follows from (\ref{qm km kp lm km kp}).
Using \eqref{another form of qm}, it follows  from (\ref{trace MS vanish}) that $\mathcal{L}\mathcal{Q}=-\mathcal{Q}\mathcal{L}$ and then \eqref{Lhat22 L21 w Q21 multiply} follows by direct computation. Since $\mathcal{Q}_{21}$  is an entire function of $k\in \C$, the expression for $\mathcal{Q}_{21}$ in \eqref{q21 formula} follows from the asymptotic formulas
\begin{align}\label{PhiAasymptotics}
\Phi(z, k^+) = \begin{pmatrix} \overline{f(z)} & 1 \\ f(z) & -1 \end{pmatrix} + O(k^{-1}),\quad 
\mathcal{A}(k) = \mathbb{I} + O(k^{-1}), \quad k \to \infty,
\end{align}
and Liouville's theorem.
\end{proof}

\subsection{The Riemann surface $\hat{\mathscr{S}}_z$ and the function $H(z,k)$} 
Let $\hat{\mathscr{S}}_z$ denote the double covering of  $\mathscr{S}_z$ obtained by adding the cut $C_{k_1}$ from $k_{1}$ to $\overline{k}_1$ on both the upper and lower sheet of $\mathscr{S}_z$. Thus a point $(k,\pm\lambda,\pm\mu)$ on $\hat{\mathscr{S}}_z$ is specified by $k\in \hat{\C}$ together with a choice of signs of $\lambda$ and $\mu$ defined in (\ref{degenerated RS}).
We specify the sheets of $\hat{\mathscr{S}}_z$ by requiring that $\lambda \to 1 $ $(\lambda \to -1)$ as $k\to \infty$ on sheets 1 and 2 (sheets 3 and 4), and by requiring that $\mu \sim k$ ($\mu \sim -k$) as $k\to \infty$ on sheets 1 and 3 (sheets 2 and 4). As $k$ crosses $[-i z,i \bar{z}]$, $\lambda$ changes sign but $\mu$ does not. As $k$ crosses $[k_{1},\overline{k}_1]$, $\mu$ changes sign but $\lambda$ does not. We define the function $H(z,k)$ by
\begin{equation}\label{H on Shat}
H(z,k) =\frac{\hat{\mathcal{L}}_{22}-\mathcal{L}_{21}\sqrt{w^{2}-1}}{\hat{\mathcal{L}}_{22}+\mathcal{L}_{21}\sqrt{w^{2}-1}},\quad k\in \hat{\mathscr{S}}_z,
\end{equation}
where the branch of $\sqrt{w^{2}-1}$ is fixed by the requirement that $\sqrt{w^{2}-1}=2i \Omega k+O(1)$ as $k\to \infty$ on sheets 1 and 3 of $\hat{\mathscr{S}}_z$.  In view of  the symmetries \eqref{qm km kp lm km kp} of $\mathcal{L}$ and $\mathcal{Q}$, we have
\begin{equation}
\hat{\mathcal{L}}_{22}(k,\lambda,\mu)=\hat{\mathcal{L}}_{22}(k,-\lambda,\mu), \quad \mathcal{L}_{21}(k,-\lambda,\mu)=-\mathcal{L}_{21}(k,\lambda,\mu),
   \end{equation}  
and therefore
\begin{equation}\label{H k la mu symmetry for -la and -mu}
H(k,\lambda,\mu)=\frac{1}{H(k,-\lambda,\mu)}=\frac{1}{H(k,\lambda,-\mu)}.
\end{equation}
It follows that $H(z,k)$ can be viewed as a single-valued function on $\Sigma_{z}$. On the upper sheet $\Sigma_{z}^+$, $H(z,k)$ is given by the values of $H(k,\lambda,\mu)$  on sheet 1 or sheet 4, while on the lower sheet $\Sigma_{z}^-$, $H(z,k)$ is given by the inverse of those values. 

For simplicity, we assume in what follows that $z$ is such that $m_1(z)$ does not lie on $\Gamma$ or on one of the branch cuts (this is the generic case; in the end, the solution can be extended to these values of $z$ by continuity).
It then follows from (\ref{q21 formula}) and (\ref{H on Shat}) that the zeros and poles of $H(z,k)$ on $\Sigma_z$ belong to the set $\{m_1^+, m_1^-\} \subset \Sigma_z$, and that either: $(i)$ $m_1^+$ is a double zero and $m_1^-$ is a double pole of $H$ or $(ii)$ $m_1^+$ is a double pole and $m_1^-$ is a double zero of $H$. Indeed, the only other possibility is that the numerator and the denominator in (\ref{H on Shat}) both have simple zeros at $m_1^+$; but then $\mathcal{L}_{21} = 0$ at $m_1^+$ and since $\mathcal{Q}_{21}$ also vanishes at $m_1$, the condition $\mathcal{Q}\mathcal{L} + \mathcal{L}\mathcal{Q} = 0$ implies that $\mathcal{L}_{11} \mathcal{Q}_{11} = 0$ at $m_1^+$, which contradicts the fact that $\det \mathcal{Q}$ and $\det \mathcal{L}$ are nonzero at $m_1^+$.
For definiteness, we will henceforth assume that case $(i)$ applies; the arguments are very similar and the final answer is the same when $(ii)$ applies.

\subsection{The scalar RH problem}
Define the scalar-valued function $\psi(z,k)$  by
\begin{equation}\label{psi def}
\psi(z,k) = \frac{\log H(z,k)}{y}, \quad k\in \Sigma_{z},
\end{equation}
where $y$ is defined in \eqref{parameterization of Sigmaz}. We view $\log H$ as a single-valued function on $\Sigma_z$ by  introducing a cut $[m_1,k_1]^- \cup [k_1,m_1]^+$ from the double pole $m_1^-$ to the double zero $m_1^+$ (see Figure \ref{from the pole m1 minus to the zero m1 plus}).
Letting $(\log H)_{+}$  and $(\log H)_{-}$ denote the boundary values of $\log H$ on the right and left sides of $[m_1,k_1]^- \cup [k_1,m_1]^+$, we have (see Figure \ref{from the pole m1 minus to the zero m1 plus})
\begin{equation}\label{logH jump across k1 mop}
(\log H)_{-}(k)=(\log H)_{+}(k)+4\pi i, \quad k\in [m_1,k_1]^- \cup [k_1,m_1]^+.
\end{equation}
Long computations using (\ref{PhiAasymptotics}) show that 
 \begin{equation}
 H(z,k^+) = f^{2}(z)+O(k^{-1}), \quad k\to \infty. 
 \end{equation}
Hence we may fix the overall branch of $\log H$ by requiring that
\begin{equation}\label{logHasymptotics}
 \log H(z,k^+) = 2\log f(z) + O(k^{-1}) \quad \mathrm{as}\quad k\to \infty,
 \end{equation} 
where the branch of $\log f$ is fixed so that $\log f(z) \to 0$ as $z \to \infty$. Then
\begin{equation}\label{well define of psi on Sgz}
\begin{split}
\log H(z,k^{+})&=-\log H(z,k^{-}),\quad y(z,k^{+})=-y(z,k^{-}),
\end{split}
\end{equation}
which shows that $\psi(z,k^+) = \psi(z,k^-)$. Thus $\psi$ descends to a function of $k \in \C$ and we can formulate a scalar RH problem for $\psi(z,k)$ as follows.

\begin{figure}[ht!]
 \centering
  \includegraphics[width=0.8\linewidth]{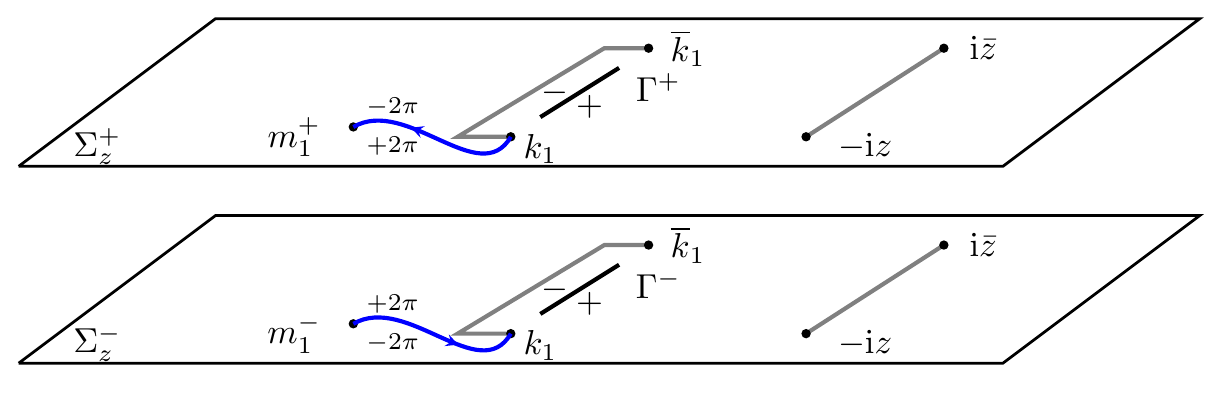}
\caption{\small{The oriented cut $[m_1,k_1]^- \cup [k_1,m_1]^+$ from the double pole $m_1^-$ on the lower sheet to the double zero  $m_1^+$ on the upper sheet of $\Sigma_z$. }}\label{from the pole m1 minus to the zero m1 plus}
\end{figure}

\begin{prop}\label{scalar RH problem}
The complex-valued function $\psi(z,\cdot)$ has the following properties: 
\begin{itemize}
\item $\psi(z,k)$ is analytic for $k\in \hat{\C}\backslash(\Gamma \cup[k_{1},m_{1}])$.
\item Across $\Gamma$, $\psi(z,k)$ satisfies the jump relation
\begin{align}\label{psijump}
\psi_{-}(z,k)=\psi_{+}(z,k)+\frac{2}{y(z,k^+)}\log \left(\frac{\sqrt{w^{2}-1}-w}{\sqrt{w^{2}-1}+w}(k^+)\right), \quad k\in \Gamma.
\end{align}

\item Across the oriented straight-line segment $[k_{1},m_{1}]$, $\psi(z,k)$ satisfies the jump relation
\begin{equation}\label{psijumpk1m1}
\psi_{-}(z,k)=\psi_{+}(z,k)+\frac{4\pi i}{y(z,k^+)}, \quad k\in [k_{1},m_{1}],
\end{equation}
where  $k^+$ is the point in the upper sheet of $\Sigma_z$ which projects onto $k$,  and $\psi_{-}(z,k)$ and $\psi_{+}(z,k)$ denote the values of $\psi$ on the left and right sides of $[k_{1},m_{1}]$.
\item As $k\to m_{1}$, 
\begin{equation}\label{k close to m1}
\psi(z,k)\to \frac{2}{y(z,k^+)}\log(k-m_{1}).
\end{equation}

\item As $k\to \infty$,
\begin{equation}\label{psi at infinity}
\psi(z,k) = \frac{2\log f}{k^{2}}+O(k^{-3}).
\end{equation}
\item As $k\to k_{1}$, 
\begin{equation}\label{k close to k1}
\psi(z,k) \to \frac{2\pi i}{y},
\end{equation}
where  $y=y(z,k^+)$  for $k^+$ just to the left of the cut $[k_1,m_1^{+}]$ on $\Sigma_z^{+}$, and is analytically continued around the endpoint $k_1$ so that it equals $y(z,k^-)$ when $k^-$ lies just to the left of the cut $[m_1^{-},k_1]$ on $\Sigma_z^{-}$.
\end{itemize} 
\end{prop}
\begin{proof}
The analyticity of $\psi$ for $k\in \hat{\C}\backslash (\Gamma \cup [k_{1},m_{1}])$ follows from the  definitions of $y(z,k)$ and $H(z,k)$.
It follows from \eqref{trace determinant of L and Q}, (\ref{QLjumps}), the fact that $\tr(\mathcal{Q}\mathcal{L}) = 0$, and tedious calculations that 
\begin{equation*}
\mathcal{L}_{21+}=(1-2w^{2})\mathcal{L}_{21-}+2w\hat{\mathcal{L}}_{22-},\quad \hat{\mathcal{L}}_{22+}=2w(w^{2}-1)\mathcal{L}_{21-}+(1-2w^{2})\hat{\mathcal{L}}_{22-},
\end{equation*}
for $k\in \Gamma^{+}$.
Hence the jump of $H(z,k)$ across $\Gamma^+$ is given by 
\begin{equation}\label{H jump on Gamma minus}
H_{-}(z,k)=H_{+}(z,k)\left(\frac{\sqrt{w^{2}-1}-w}{\sqrt{w^{2}-1}+w}\right)^2, \quad k\in \Gamma^+,
\end{equation}
which gives (\ref{psijump}). The jump (\ref{psijumpk1m1}) follows from (\ref{logH jump across k1 mop}).
The asymptotic behavior \eqref{k close to m1} is a consequence of the fact $H(z,k)$ has a double zero at $k=m_{1}^{+}$. 
Equation \eqref{psi at infinity} follows from (\ref{logHasymptotics}) and the fact that $y(z,k^+) = k^2 +O(k)$ as $k\to \infty$.
The behavior of $\psi$ as $k\to k_1$ follows from \eqref{H k la mu symmetry for -la and -mu} and (\ref{logH jump across k1 mop}).
\end{proof}

\subsection{Solution of the scalar RH problem}\label{solve scalar RH problem}
Using the Sokhotski-Plemelj formula, we find that the solution of the scalar RH problem presented in Proposition \ref{scalar RH problem} is given by
\begin{align}\nonumber
\psi(z,k)=&\; \frac{1}{\pi i}\int_{\Gamma}\frac{\diff k^{\prime}}{y(z,k^{+\prime})(k^{\prime}-k)}\log\left(\frac{\sqrt{w^{2}-1}-w}{\sqrt{w^{2}-1}+w}(k^{+\prime})\right)
	\\\label{solution of auxilary RH}
& +2\int_{[k_{1},m_{1}]}\frac{\diff k^{\prime}}{y(z,k^{\prime+})(k^{\prime}-k)}.
\end{align}
This can be rewritten in terms of contour integrals on $\Sigma_{z}$ as 
\begin{align}\label{divisor form solution of auxilary RH}
\psi(z,k)&= 2\int_{\Gamma^{+}}\frac{h(k^\prime) \diff k^{\prime}}{y(z,k^{\prime})(k^{\prime}-k)}
-2\int_{k_{1}}^{' m^{-}_{1}}\frac{\diff k^{\prime}}{y(z,k^{\prime})(k^{\prime}-k)},
\end{align}
where $h(k)$ is given by
$$h(k) = \frac{1}{2\pi i}\log\left(\frac{\sqrt{w^{2}(k)-1}-w(k)}{\sqrt{w^{2}(k)-1}+w(k)}\right), \qquad k \in \Gamma^+,$$
and the prime on the integral from $k_1$ to $m_1^-$ indicates that the path of integration does not necessarily lie in the complement of the cut basis $\{a, b\}$.
Since $w(k) = -2ik\Omega$, $h(k)$ can be written as in (\ref{def h}).
Letting $k\to \infty$ in (\ref{divisor form solution of auxilary RH}) and recalling \eqref{psi at infinity}, we obtain 
\begin{subequations}\label{logfboth}
\begin{align}
\log f&=\int_{k_1}^{'m^{-}_{1}}\frac{k\diff k}{y}-\int_{\Gamma^{+}}h(k)\frac{k\diff k}{y},\label{high order integral}
	\\ 
\int_{k_1}^{'m^{-}_{1}}\frac{\diff k}{y}&=\int_{\Gamma^{+}}h(k)\frac{\diff k}{y}\label{lower order integral}.
\end{align}
\end{subequations}

\section{Theta functions}\label{thetasec}
In this section, we derive the expressions for $f$, $e^{2U}$, and $a$ given in \eqref{Ernst potential for f} and  \eqref{e2U and a}.

We define the map $\phi:\Sigma_{z}\to \C$ by 
\[
\phi(k) \equiv \phi(z,k) = \int_{-i z}^{k}\omega, \quad k\in \Sigma_{z},
\]
where it is assumed that the integration contour lies within the fundamental polygon determined by $\{a,b\}$. 
The Jacobian variety $\text{Jac}(\Sigma_{z})$ of $\Sigma_z$ is defined by $\text{Jac}(\Sigma_{z}) = \C/\mathbb{L}$, where $\mathbb{L}$ denotes the discrete lattice generated by $1$ and $B$. The composition of $\phi$ with the projection $\C\to \text{Jac}(\Sigma_{z})$ is the Abel map with base point $-i z$, and $\phi(k_1)\in \C$ projects to the vector of Riemann constants in $\text{Jac}(\Sigma_{z})$, see  \cite[Chap. VII]{farkas1992riemann}.

\subsection{Proof of expression \eqref{Ernst potential for f} for $f$}
We first suppose that the contour from $k_1$ to $m_1^-$ in (\ref{logfboth}) lies in the fundamental polygon determined by $\{a, b\}$. Utilizing \eqref{def U and I}, \eqref{omega A eta}, and \eqref{Abelian differential third Sigmaz} in \eqref{logfboth}, we  deduce that
\begin{equation}\label{u ms ks}
u=\int_{k_1}^{m_1^-} \omega,\quad
f=e^{-\int_{k_1}^{m_1^-} \omega_{\infty^{+}\infty^{-}}+I}.
\end{equation}
 In view of the properties of zero divisors of theta functions (see \cite[Thm. VI.3.1]{farkas1992riemann}) and the following symmetry properties for Abelian differentials (see \cite[Chap. III]{farkas1992riemann})
\begin{equation}\label{identities for abel differentials}
\int_{b}\omega_{PQ}=2\pi i \int_{Q}^{P}\omega,\quad
\int^{P}_{Q}\omega_{\infty^{+}\infty^{-}}=\int^{\infty^{+}}_{\infty^{-}}\omega_{PQ}
, \quad P,Q\in \Sigma_{z},
\end{equation}
we obtain
\begin{equation}\label{eqiv 1 and 2 at P-inftyplus}
e^{- \int_{k_1}^{m_1^-} \omega_{\infty^{+}\infty^{-}}}=\frac{\Theta\left(\phi(\infty^{+})\right)\Theta\left(\phi(\infty^{-})-u\right)}{\Theta\left(\phi(\infty^{+})-u\right)\Theta\left(\phi(\infty^{-})\right)}.
\end{equation}
Since $\phi(\infty^{+})=-\phi(\infty^{-})$ and the theta function $\Theta$ is even, we can rewrite \eqref{eqiv 1 and 2 at P-inftyplus} as
\begin{equation*}
e^{- \int_{k_1}^{m_1^-}\omega_{\infty^{+}\infty^{-}}}=\frac{\Theta\left(u-\phi(\infty^{-})\right)}{\Theta\left(u+\phi(\infty^{-})\right)},
\end{equation*}
which together with \eqref{u ms ks} leads to \eqref{Ernst potential for f}.
It is easy to see that the answer remains if invariant if the contour from $k_1$ to $m_1^-$ is replaced by a contour which does not lie in the fundamental polygon determined by $\{a, b\}$, cf. \cite{lenells2011boundary}.

\subsection{Proof of expressions \eqref{e2U and a} for $e^{2U}$ and $a$}
Since $u\in i \R$ and $I\in \R$, the expression for $e^{2U}$ in \eqref{e2U and a} can be derived as in \cite{lenells2011boundary}. 
With the expression for $e^{2U}$ at hand, the expression for $a$ in  \eqref{e2U and a} can be obtained by following the argument in \cite[Sect. V]{klein1998physically}.

\section{The metric function $e^{2\kappa}$}\label{kappasec}
One useful tool in the study of the Ernst equation (\ref{Ernst equation}) is the use of branch point condensation arguments \cite{korotkin2000theta}. In this section, we apply such arguments to derive the expression \eqref{e2kappa} for $e^{2\kappa}$. An analogous derivation was considered in \cite{lenells2011boundary}. However, in contrast to the situation in \cite{lenells2011boundary}, our function $h(k)$ which determines the jump of the scalar RH problem for $\psi(z,k)$ does not vanish at the endpoints of $\Gamma$.
This means that the condensation argument has to be modified. 

\subsection{Proof of expression \eqref{e2kappa} for $e^{2\kappa}$}
Given an integer $g\geq 2$ and $g-1$ branch cuts $\{[E_{j},F_{j}]\}_{j=2}^{g}$, let $\hat{\Sigma}_{z}$ be the Riemann surface of genus $g$ defined by the equation
\[
\hat{y}^2 = (k-\xi)(k-\bar{\xi})(k-k_1)(k - \bar{k}_1)\prod_{j=2}^{g}(k-E_j)(k-F_j),
\]
where $\xi := -i z$. We define a cut basis $\{\hat{a}_j,\hat{b}_j\}_{j=1}^{g}$ on $\hat{\Sigma}_{z}$ as follows: $\hat{a}_1$ surrounds the cut $C_{k_1}$ and, for $j = 2, \dots, g$, $\hat{a}_j$ surrounds the cut $[E_{j}, F_{j}]$ in the counterclockwise direction; $\hat{b}_j$ enters the upper sheet on the right side of $[-i z,i \bar{z}]$ and exits again on the right side $C_{k_1}$ for $j = 1$ and on the right side of $[E_{j},F_{j}]$ for $j = 2, \dots, g$. Then $\{\hat{a}_j,\hat{b}_j\}_{j=1}^{g}$ is a natural generalization of the basis $\{a,b\}$ on $\Sigma_z$.
Let $\hat{\omega}=\{\hat{\omega}_1, ..., \hat{\omega}_g\}^{T}$ denote the canonical dual basis and let $\hat{\Theta}(\hat{\omega}) \equiv \Theta(\hat{\omega}|\hat{B})$ be the associated  theta function. Let $p,q\in \C^{q}$ be $z$-independent vectors which satisfy $\hat{B} p+q\in \R^{g}$. The theta function with characteristics $p,q\in \R^{g}$ is defined for $\hat{v} \in \C^g$ by 
\[
\hat{\Theta}\begin{bmatrix}
p\\q
\end{bmatrix}(\hat{v}) = \hat{\Theta}(\hat{v}+\hat{B}p+q)\exp\lp 2\pi i \left(\frac{1}{2}p^{T}\hat{B}p+p^{T}(\hat{v}+q)\right)\rp.
\]
Then the function $\hat{f}(z)$ defined by
\begin{align}\label{fhat kappa general riemann surface}
\hat{f}(z) =\frac{\hat{\Theta}\begin{bmatrix}
p\\q
\end{bmatrix}\left(\int_{\xi}^{\infty^{+}}\hat{\omega}\right)}{\hat{\Theta}\begin{bmatrix}
p\\q
\end{bmatrix}\left(\int_{\xi}^{\infty^{-}}\hat{\omega}\right)}
\end{align}
satisfies the Ernst equation (\ref{Ernst equation}) and the corresponding metric function $e^{2\kappa}$ is given by
\begin{align}\label{e2kappahat}
e^{2\hat{\kappa}}=\hat{K}_0\frac{\hat{\Theta}\begin{bmatrix}
p\\q
\end{bmatrix}(0)\hat{\Theta}\begin{bmatrix}
p\\q
\end{bmatrix}\left(\int_{\xi}^{\bar{\xi}}\hat{\omega}\right)}{\hat{\Theta}(0)\hat{\Theta}\left(\int_{\xi}^{\bar{\xi}}\hat{\omega}\right)},
\end{align}
where $\hat{K}_0$ is a constant determined by the requirement that $e^{2\hat{\kappa}}=1$ on the rotation axis \cite{korotkin2000theta}.
Our goal is to recover the Ernst potential (\ref{Ernst potential for f}) by letting the branch points $E_j, F_j$ in (\ref{fhat kappa general riemann surface}) condense along the contour $\Gamma$ for an appropriate choice of the characteristics $p,q$. Then, by applying the same condensation to (\ref{e2kappahat}), we will obtain the expression (\ref{e2kappa}) for $e^{2\kappa}$.

In order to handle the fact that $h(k)$ does not vanish at the endpoints of $\Gamma$, we define, for each integer $n \geq 1$,  an extension $\Gamma_n$ of $\Gamma$ by $\Gamma_n = [-i \rho_0-\frac{i}{n}, i \rho_0 + \frac{i}{n}]$. We first let each branch cut $[E_{j},F_{j}]$ shrink to a point $\kappa_j\in \Gamma_n$. In this limit, we have
\begin{subequations}\label{shrink property}
\begin{align}
\left(\hat{\omega}_{1}, \dots,\hat{\omega}_{g}\right)&\rightarrow \big(\omega, \frac{1}{2\pi i}\omega_{\kappa_2^{+}\kappa_2^{-}},...,\frac{1}{2\pi i}\omega_{\kappa_{g}^{+}\kappa_{g}^{-}}\big);\\
\hat{B}_{11}&\to B; \quad \hat{B}_{1j}\to \int_{\kappa_{j}^{-}}^{\kappa_{j}^{+}}\omega, \quad j=2,..., g;
	\\
\hat{B}_{ij}&\to \frac{1}{2\pi i}\int_{\kappa_{j}^{-}}^{\kappa_{j}^{+}}\omega_{\kappa_{i}^{+}\kappa_{i}^{-}}, \quad i\neq j, \quad i, j=2,..., g;
	\\
\hat{B}_{jj}&= \frac{1}{\pi i}\log|E_j-F_j|+O(1), \quad  j=2,..., g.
\end{align}
\end{subequations}
We then let the $\kappa_{j}$ condense along $\Gamma_n$ with the density determined by the measure 
\beq\label{condense measure}
\diff m_n(\kappa)=-\frac{1}{2}\frac{\diff h_n}{\diff\kappa}(\kappa)\diff\kappa,\quad \kappa\in\Gamma_n,
\eeq
where $h_n:\Gamma_n\to \C$ is a sequence of smooth functions which vanish identically near the endpoints of $\Gamma_n$ and which converge pointwise to $h(k)$ on $\Gamma$ as $n \to \infty$. 
Choosing
\[p=(p_1,p_2,..., p_{g})=:(p_1,\tilde{p})\in \C\times\R^{g-1},\quad q= (0,...,0) \in \R^{g},
\] 
where the vector $\tilde{p}\in \R^{g-1}$ is such that $\tilde{p}_{j}\in (0,\frac{1}{2})$, $j=2,..., g$, we find from \eqref{shrink property} that, as the branch points condense along $\Gamma_n$,
\begin{align}\label{thetahat expression}
& \hat{\Theta}\! \begin{bmatrix} p\\ 0 \end{bmatrix}\! \! \lp\int^{Q}_{P}\hat{\omega}\rp \to \Theta \! \begin{bmatrix} p_1 \\ 0 \end{bmatrix} \!\! \lp \int_{\Gamma_n}\diff m_n(\kappa)\int_{\kappa^{-}}^{\kappa^{+}}\omega+ \int^{Q}_{P}\omega\rp 
e^{\frac{L_n}{2}+\int_{\Gamma_n}\diff m_n(\kappa)\int_{P}^{Q}\omega_{\kappa^{+}\kappa^{-}}},
\end{align}
where 
\begin{align}\nonumber
L_n & = \int_{\Gamma_n}\diff m_n(\kappa_1)\int_{\Gamma_n}^{\prime} \diff m_n(\kappa_2)\int^{\kappa_2^{+}}_{\kappa_2^{-}}\omega_{\kappa_1^{+}\kappa_1^{-}}(\kappa)
	\\ \label{Lndef}
& =\frac{1}{2}\int_{\Gamma_n} \diff\kappa_1 h_n(\kappa_1)\int_{\Gamma_n}^{\prime} h_{n}(\kappa_2)\frac{\partial\omega_{\kappa^{+}_{1}\kappa^{-}_{1}}}{\partial\kappa_1}(\kappa_2^{+}),
\end{align}
and the prime on the integral indicates that the contour $\Gamma_n$ should be deformed slightly so that the pole at $\kappa_1=\kappa_2$ is avoided.
Integrating by parts and using \eqref{identities for abel differentials}, we obtain
\begin{align}\label{def check u}
\int_{\Gamma_n}dm_n(\kappa)\int_{\kappa^{-}}^{\kappa^{+}}\omega&=\int_{\Gamma_n^+}h_n(\kappa)\omega(\kappa)=:u_n,
	\\
\int_{\Gamma_n}\diff m_n(\kappa)\int^{\infty^+}_{\infty^-}\omega_{\kappa^{+}\kappa^{-}}&=\int_{\Gamma_n^{+}}h_n(\kappa)\omega_{\infty^{+}\infty^{-}}(\kappa)=:I_n. \label{def check I}
\end{align}
Using \eqref{thetahat expression}-\eqref{def check I} in \eqref{fhat kappa general riemann surface}, we conclude that the right-hand side of (\ref{fhat kappa general riemann surface}) converges to the following Ernst potential in the limit as the branch points condense along $\Gamma_n$:
\beq\label{fhat degeneration method formula}
f_n=\frac{\Theta \begin{bmatrix} p_1 \\ 0 \end{bmatrix} \lp  u_n+\int^{\infty^+}_{\xi}\omega\rp}{\Theta \begin{bmatrix} p_1 \\ 0 \end{bmatrix} \lp  u_n+\int^{\infty^-}_{\xi}\omega\rp}e^{I_n}.
\eeq

Letting $p_1 =0$ and taking $n \to \infty$ in (\ref{fhat degeneration method formula}), we recover the Ernst potential (\ref{Ernst potential for f}). We would like to take the same limit in the expression for the metric function $e^{2\kappa_n}$ associated with $f_n$. 
However, we first need to regularize the expression for $L_n$.
We find from \eqref{Abelian differential third Sigmaz} that
\begin{align*}
\frac{\partial \omega_{\kappa_1^+ \kappa_1^-}}{\partial \kappa_1} (\kappa_2^+)
= \frac{y'(\kappa_1^+)(\kappa_2-\kappa_1) + y(\kappa_1^+)}{(\kappa_2-\kappa_1)^2y(\kappa_2^+)} d\kappa_2
- \bigg(\int_a \frac{y'(\kappa_1^+)(\kappa-\kappa_1) + y(\kappa_1^+)}{(\kappa-\kappa_1)^2 y(\kappa)}d\kappa \bigg) \omega(\kappa_2^+).
\end{align*}
Hence we can rewrite (\ref{Lndef}) as
\begin{align}\label{Lnsplit}
L_n =  L_n^{reg}+L_n^{rem},
\end{align}
where the regularized term $L_n^{reg}$ is defined by
\begin{align*}\nonumber
L_n^{reg}
= &\; \frac{1}{2}\int_{\Gamma_n} \diff\kappa_1 h_n(\kappa_1)\int_{\Gamma_n}^{\prime} h_{n}(\kappa_2) \bigg( \frac{\partial\omega_{\kappa^{+}_{1}\kappa^{-}_{1}}}{\partial\kappa_1}(\kappa_2^{+}) - \frac{\diff\kappa_2}{(\kappa_2-\kappa_1)^2}\bigg)
	\\ \nonumber
= &\; \frac{1}{2} \bigg\{\int_{\Gamma_n} d\kappa_1 h_n(\kappa_1) \int_{\Gamma_n}' h_n(\kappa_2) 
\bigg(\frac{y'(\kappa_1^+)(\kappa_2-\kappa_1) + y(\kappa_1^+)-y(\kappa_2^+)}{(\kappa_2-\kappa_1)^2y(\kappa_2^+)} \bigg)d\kappa_2
	\\
&- \bigg(\int_{\Gamma_n} d\kappa_1 h_n(\kappa_1) \int_{a} \frac{y'(\kappa_1^+)(\kappa-\kappa_1) + y(\kappa_1^+)}{(\kappa-\kappa_1)^2 y(\kappa)}d\kappa \bigg)
\bigg(\int_{\Gamma_n} h_n(\kappa_2) \omega(\kappa_2^+)\bigg)
\bigg\}
\end{align*}
and the remainder $L_n^{rem}$ is given by
\begin{align*}
L_n^{rem} = \frac{1}{2} \int_{\Gamma_n} d\kappa_1 h_n(\kappa_1) \int_{\Gamma_n}'  h_n(\kappa_2) 
\frac{d\kappa_2}{(\kappa_2-\kappa_1)^2} .
\end{align*}
We have arranged the definition of $L_n^{reg}$ so that the integrand only has a simple pole at $\kappa_2 = \kappa_1$. It follows that $L_n^{reg}$ has the well-defined limit $L^{reg}$ defined in (\ref{L def}) as $n \to \infty$. The key point is that the first double integral in (\ref{L def}) converges at the endpoints of $\Gamma$ despite the fact that $h$ is nonzero at these endpoints. The term $L_n^{rem}$, on the other hand, diverges as $n \to \infty$ due to ever growing contributions from the endpoints of $\Gamma$. However, since $L_n^{rem}$ is independent of $z$, it can be absorbed into the coefficient $\hat{K}_0$.

Furthermore, there exists a constant $C$ such that
\beq 
e^{-\int^{P}_{P_0}\omega_{\xi\bar{\xi}}}=C\sqrt{\frac{P-\bar{\xi}}{P-\xi}}, \quad P\in \Sigma_{z},
\eeq
because both sides have a simple pole at $\xi $, a simple zero at $\bar{\xi}$, and are analytic elsewhere on $\Sigma_{z}$. Hence, 
\beq
\int_{\xi}^{\bar{\xi}}\omega_{\kappa^{+}\kappa^{-}}=-\int_{\kappa^{-}}^{\kappa^{+}}\omega_{\xi\bar{\xi}}=\left[\log\sqrt{\frac{P-\bar{\xi}}{P-\xi}}\right]_{\kappa^{-}}^{\kappa^{+}} \in \pi i+2\pi i \Z.
\eeq
Consequently, the second term in the exponent in \eqref{thetahat expression} is independent of $z$ and can also be absorbed into the coefficient $\hat{K}_0$. 

The above absorptions can be carried out by defining a new constant $K_{0n}$ by
\begin{equation}\label{K0n}
K_{0n}=\hat{K}_0 e^{L_{n}^{rem}+\int_{\Gamma_n}\diff m_n(\kappa)\int_{\xi}^{\bar{\xi}}\omega_{\kappa^{+}\kappa^{-}}}. 
\end{equation}
As the branch points condense, we have $\hat{\Theta}(v) \to \Theta(v_1)$ for each $v \in \C^g$, where $v_1$ denotes the first component of the vector $v \in \C^g$. Using this fact together with \eqref{thetahat expression}, \eqref{Lnsplit}, and \eqref{K0n} to compute the limit of \eqref{e2kappahat} as the branch points condense along $\Gamma_n$, we see that the metric function $e^{2\kappa_n}$ corresponding to the Ernst potential $f_n$ can be expressed by 
\beq\label{kappa Lreg}
e^{2\kappa_n}=K_{0n}\frac{\Theta\!\begin{bmatrix}
p_1 \\ 0
\end{bmatrix}\!\!\lp u_n\rp\!\Theta\!\begin{bmatrix}
p_1 \\ 0
\end{bmatrix}\!\!\lp u_n+ \int_{\xi}^{\bar{\xi}}\omega\rp}{\Theta(0)\Theta\left(\int_{\xi}^{\bar{\xi}}\omega\right)}e^{ L_n^{reg}}.
\eeq
We can now let $p_1=0$ and take the limit $n\to \infty$ in \eqref{fhat degeneration method formula} and \eqref{kappa Lreg}. Then $f_n$ converges to the solution (\ref{Ernst potential for f}) of the BVP (\ref{boundary value problem}), and $e^{2\kappa_n}$ converges to the associated metric function $e^{2\kappa}$ given by
\beq\label{kappa Lreg limit}
e^{2\kappa}=K_{0}\frac{\Theta\lp u\rp\!\Theta\lp u+ \int_{\xi}^{\bar{\xi}}\omega\rp}{\Theta(0)\Theta\left(\int_{\xi}^{\bar{\xi}}\omega\right)}e^{ L^{reg}},
\eeq
where $K_0$ is a $z$-independent constant. This completes the proof of the expression (\ref{e2kappa}) for $e^{2\kappa}$; the explicit expression (\ref{K0 formula}) for the constant $K_0$ will be derived in Section \ref{axissec}.

\section{Solution near the rotation axis}\label{axissec}
In this section, we first complete the proof Theorem \ref{mainth2} by studying the asymptotic behavior of the Ernst potential and the metric functions near the rotation axis $\rho = 0$. We then use these results to establish expressions for $a_0$ and $K_0$, thus completing also the proof of Theorem \ref{mainth1}.

As $\rho\to 0$, the branch cut $[-i z, i \bar{z}]$ shrinks to a point and  $\Sigma_{z}$ degenerates to the genus zero Riemann surface $\Sigma^{\prime}$ defined by \eqref{degenerated RS}.
In order to study this degeneration, we introduce the axis-adapted cut basis $\{\tilde{a},\tilde{b}\}$ on $\Sigma_{z}$ by (see Figure \ref{axis-adapted homology basis})
\begin{equation}
\begin{pmatrix}
\tilde{a} \\
\tilde{b}
\end{pmatrix} =
\begin{pmatrix}
-1 & 0 \\
0 & -1
\end{pmatrix}
\begin{pmatrix}
a \\
b
\end{pmatrix}.
\end{equation}
\begin{figure}[ht!]
  \centering
   \vspace{-.5cm}
  \includegraphics[width=0.35\linewidth]{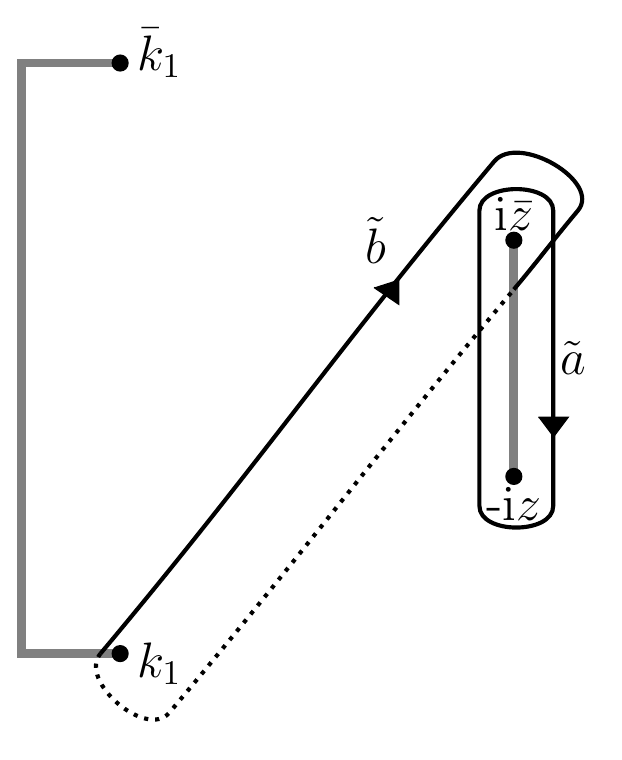}
  \vspace{-0.2cm}
\caption{The axis-adapted basis $\{\tilde{a}, \tilde{b}\}$ on $\Sigma_z$.}
\label{axis-adapted homology basis}
\end{figure}
It then follows from the transformation formula for theta functions \cite[Eq. (12)]{fay2006theta} that there exists a constant $c_{0}$ such that 
\begin{equation}\label{linearity of theta function under basis transformation}
\Theta(\tilde{{v}}|\tilde{B})=c_0 \Theta(v|B),
\end{equation}
where $\tilde{v}�= -v$ and $\tilde{B}$ denotes the period matrix associated with $\{\tilde{a},\tilde{b}\}$. Let $\tilde{\omega}$ and $\tilde{u}$ denote the analogs of $\omega$ and $u$ in the axis-adapted basis $\{\tilde{a},\tilde{b}\}$ on $\Sigma_{z}$. Then
\begin{equation}\label{def tilde u and tilde I}
\tilde{\omega}=-\omega,\quad \tilde{u}=\int_{\Gamma^{+}}h\tilde{\omega}=-u.
\end{equation}
It can also be verified by a residue computation that  
\begin{equation*}
\omega_{\infty^{+}\infty^{-}}=\tilde{\omega}_{\infty^{+}\infty^{-}}-2\pi i \tilde{\omega}
\end{equation*}
and consequently
\begin{equation}\label{I and I tilde}
I=\tilde{I}-2\pi i \tilde{u} \quad \text{where} \quad \tilde{I} = \int_{\Gamma^+}h \tilde{\omega}_{\infty^{+}\infty^{-}}.
\end{equation}
Using the notation $\tilde{\Theta}(v) \equiv \Theta(v|\tilde{B})$, we can rewrite the expression \eqref{Ernst potential for f} for the Ernst potential as 
\begin{equation}\label{f in terms of I tilde}
f(z)=\frac{\tilde{\Theta}(\tilde{u}-\int_{-i z}^{\infty^{-}}\tilde{\omega})}{\tilde{\Theta}(\tilde{u}+\int_{-i z}^{\infty^{-}}\tilde{\omega})}e^{\tilde{I}-2\pi i \tilde{u}}.
\end{equation}
Define $c \equiv c(z) \in \C$, $J' \equiv J'(\zeta) \in \R$, and $K' \equiv K'(\zeta) \in \C$ by
\begin{equation*}
c=\int^{\infty^{-}}_{k_{1}}\tilde{\omega}, \quad  
J^{\prime} = \int_{\Gamma^{+}}h\omega^{\prime}_{\zeta^{+}\zeta^{-}}, \quad
 K^{\prime} = \int_{k_1}^{\infty^{-}}\omega^{\prime}_{\zeta^{+}\zeta^{-}},
\end{equation*}
where $\omega^{\prime}_{PQ}$ denotes the Abelian differential of the third kind on $\Sigma^{\prime}$.
Recalling the expression (\ref{Abelian differential third Sigmazprime}) for $\omega^{\prime}_{\zeta^{+}\zeta^{-}}$, we see that $J'$ can be written as in (\ref{Jprimedef}).

\begin{lemma}\label{Kprimedlemma}
  We have $-e^{-K'} = d$, where $d(\zeta)$ is the function in (\ref{Jprimedef}).
\end{lemma}
\begin{proof}
A direct computation shows that
$$\frac{d}{dk}\log \bigg(\frac{\mu(k^+) \mu(\zeta^+) 
-\zeta k+k_1^2}{k-\zeta}\bigg)
 = \frac{\mu(\zeta^+)}{(k-\zeta)\mu(k^+)}.$$
 Hence, using the expression (\ref{Abelian differential third Sigmazprime}) for $\omega^{\prime}_{\zeta^{+}\zeta^{-}}$,
\begin{align*}
K' & = -\mu(\zeta^+) \int_{k_1}^{\infty}  \frac{dk}{(k - \zeta) \mu(k^+)}
= -\log\bigg(\frac{\mu(k^+) \mu(\zeta^+) - \zeta k+k_1^2}{k-\zeta }\bigg)\bigg|_{k=k_1}^{\infty} 
	\\
& = -\log(\mu(\zeta^+) - \zeta )
+ \log\bigg(\frac{- \zeta k_1+k_1^2}{k_1-\zeta}\bigg)
= -\log\bigg(\frac{\mu(\zeta^+) - \zeta}{k_1}\bigg),
\end{align*}
which implies $e^{-K'} = (\mu(\zeta^+) - \zeta)/k_1 = -d(\zeta)$.
\end{proof}

The following lemma gives the behavior of several quantities near the rotation axis. 

\begin{lemma}\label{axislemma}
Let $\zeta > 0$. As $\rho\to 0$, it holds that
\begin{align}\nonumber
& \tilde{u} =\frac{J^{\prime}}{2\pi i}+O(\rho^{2}), && \tilde{\Theta}\bigg(\tilde{u} -\int_{-i z}^{\infty^{-}}\tilde{\omega}\bigg) =1-e^{J^{\prime}-K^{\prime}}+O(\rho^{2}), 
	\\\nonumber
& c=\frac{K^{\prime}}{2\pi i}+O(\rho^{2}),  && \tilde{\Theta}\bigg(\tilde{u} +\int_{-i z}^{\infty^{-}}\tilde{\omega}\bigg)=1-e^{-J^{\prime}-K^{\prime}}+O(\rho^{2}),
	\\\nonumber
&   \tilde{\Theta}\bigg(\tilde{u} +\int_{-i z}^{i\bar{z}}\tilde{\omega}\bigg)=1+\beta\rho+O(\rho^{2}),
&&\tilde{\Theta}\bigg(\tilde{u} +\int_{i \bar{z}}^{\infty^{-}}\tilde{\omega}\bigg) =1+e^{-J^{\prime}-K^{\prime}}+O(\rho^{2}), 
	\\\label{axisexpansions}
& \tilde{\Theta}(\tilde{u}) =1-\beta\rho+O(\rho^{2}), 
&& \hspace{-1.8cm}
 \tilde{\Theta}\bigg(\tilde{u} +\int_{-i z}^{\infty^{-}}\tilde{\omega}+\int_{i \bar{z}}^{\infty^{-}}\tilde{\omega}\bigg)=-\frac{e^{-J^{\prime}-2K^{\prime}-M^{\prime}}}{\rho} +O(1), 
\end{align} 
and
\begin{align}\label{QtildeIexpansions}
 Q(u) =1-e^{-2J^{\prime}-2K^{\prime}}+O(\rho^{2}), \quad \tilde{I} = O(\rho^2),
\end{align} 
where $\beta\in \C$ is a constant and $M^{\prime}$ is defined by
$$M^{\prime}=\frac{1}{2}\lim_{\epsilon\to 0}\left(\int^{(\zeta-\epsilon)^{+}}_{(\zeta-\epsilon)^{-}}\omega^{\prime}_{\zeta^{+}\zeta^{-}}-2\log \epsilon -2\log 2 -\pi i\right).
$$
The expansions in (\ref{axisexpansions}) and (\ref{QtildeIexpansions}) remain valid if $\tilde{u}$ and $J'$ are replaced by $0$ in all places.
\end{lemma}
\begin{proof}
According to \cite[Chap. III]{fay2006theta}, $\tilde{\omega}$ and $\tilde{B}$ admit the following expansions as $\rho\to 0$:
\begin{equation}\label{tilde u as rho to 0}
\tilde{\omega} = \frac{1}{2\pi i }\omega^{\prime}_{\zeta^{+}\zeta^{-}}+O(\rho^{2}),\quad \tilde{B}=\frac{1}{\pi i} \log \rho +\frac{M^{\prime}}{\pi i} + O(\rho^{2}).
\end{equation}
The expansions of $\tilde{u}$ and $c$ follow directly from (\ref{tilde u as rho to 0}) and the definitions. Moreover, note that
\begin{equation}\label{integral to c r s}
\int^{\infty^{-}}_{-i z}\tilde{\omega}=c+\frac{\tilde{B}}{2}-\frac{1}{2},
\quad
\int_{i\bar{z}}^{\infty^-}\tilde{\omega}=c+\frac{\tilde{B}}{2}.
\end{equation}
By using \eqref{tilde u as rho to 0}, \eqref{integral to c r s}, and the expansions of $\tilde{u}$ and $c$, we get 
$$\tilde{\Theta}\bigg(\tilde{u} -\int_{-i z}^{\infty^{-}}\tilde{\omega}\bigg)
=\sum_{N\in\{0,1\}} e^{2\pi i (\frac{1}{2}\tilde{B}N(N-1)+N(\tilde{u}-c+\frac{1}{2}))} 
+ O(\rho^2)
=1-e^{J^{\prime}-K^{\prime}} +O(\rho^2).$$
This establishes the expansion of $\tilde{\Theta}(\tilde{u} -\int_{-i z}^{\infty^{-}}\tilde{\omega})$; the other expansions in (\ref{axisexpansions}) are proved in a similar way. 
Note that \eqref{linearity of theta function under basis transformation} implies that $Q(u)$ is invariant under the change of basis $\{a,b\}\rightarrow \{\tilde{a},\tilde{b}\}$.  The asymptotic behavior of $Q(u)$ then follows by substituting the expansions in (\ref{axisexpansions}) into \eqref{def Q}.
Finally, we have $\tilde{I} = I' + O(\rho^2)$ where $I^{\prime}=\int_{\Gamma^{+}}h\omega^{\prime}_{\infty^{+}\infty^{-}}$. But $I' = 0$ because $h(k)$ is an odd function of $k$.  This proves (\ref{QtildeIexpansions}).
\end{proof}

\begin{proof}[Proof of Theorem \ref{mainth2}.]
In view of Lemma \ref{axislemma}, the behavior \eqref{fnearaxis} of $f$ near the rotation axis follows from \eqref{I and I tilde} and \eqref{f in terms of I tilde}.
The asymptotics \eqref{e2U a kappa asymptotic near rotation axis} of $e^{2U}$ follows from Lemma \ref{axislemma} and the identity $e^{2U(z)}=Q(0)e^{I}/Q(u)$. 
These computations lead to the expressions 
\begin{align}\label{fe2Uonaxis}
f(i \zeta)=\frac{1-e^{J^{\prime}-K^{\prime}}}{1-e^{-J^{\prime}-K^{\prime}}}e^{-J^{\prime}}, \quad
e^{2U(i \zeta)}=\frac{1-e^{-2K^{\prime}}}{1-e^{-2J^{\prime}-2K^{\prime}}}e^{-J^{\prime}},
\end{align}
for $f(i \zeta)$ and $e^{2U(i \zeta)}$. Using Lemma \ref{Kprimedlemma}, we see that the expressions in (\ref{fe2Uonaxis}) are consistent with the expressions (\ref{axis value of f}) and (\ref{e2U axis value}) obtained in Section \ref{laxsec}.

The fact that $a = O(\rho^2)$ as $\rho\to 0$ is a consequence of \eqref{b relation 1}, \eqref{e2U and a}, and Lemma \ref{axislemma}. Similarly, in view of \eqref{azkappaz} and \eqref{e2kappa}, it is clear that $e^{2\kappa}-1 = O(\rho^2)$. This completes the proof of Theorem \ref{mainth2}.
\end{proof}

\begin{proof}[Proof of the expressions for $a_0$ and $K_{0}$]
Expressing \eqref{e2U and a} in the axis-adapted basis $\{\tilde{a},\tilde{b}\}$ and letting $\rho \to 0$ in the resulting equation, we find
$$a_0=\lim_{\rho\to 0}\frac{\rho}{Q(0)^2}\frac{\tilde{\Theta}(\tilde{u} +\int_{-i z}^{\infty^{-}}\tilde{\omega}+\int_{i \bar{z}}^{\infty^{-}}\tilde{\omega})}{\tilde{\Theta}(\tilde{u}+ \int_{-i z}^{i\bar{z}}\tilde{\omega})}e^{2\pi i \tilde{u}-\tilde{I}}.$$
Using Lemma \ref{axislemma} we conclude that 
$$a_0 = -\frac{e^{-2K' - M'}}{(1 - e^{-2K'})^2}.$$
Letting $\zeta \to \infty$ and using that $\lim_{\zeta\to \infty} e^{-2K'} = 0$ and
\begin{equation*}
\lim_{\zeta\to \infty}(M^{\prime}+2K^{\prime})=\frac{1}{2}\lim_{R\to \infty}\left(-\int^{R^{+}}_{R^{-}}\omega^{\prime}_{\infty^{+}\infty^{-}}-2\log R-2\log 2-\pi i\right),
\end{equation*}
we infer that
$$a_{0}=-2\mathrm{i} \lim_{R\to \infty}R \exp\lp\int_{k_{1}}^{R^{+}}\omega^{\prime}_{\infty^{+}\infty^{-}}\rp
 = -2i\lim_{R \to \infty} Re^{-\int_{k_1}^{R} \frac{dk}{\mu(k^+)}}.$$
Since
$$\frac{d}{dk} \log\big(k + \mu(k^+)\big) = \frac{1}{\mu(k^+)},$$
we obtain the sought-after expression for $a_0$:
\begin{align*}
a_0 & = -2i\lim_{R \to \infty}Re^{-\log(k + \mu(k^+))\big|_{k_1}^{R}}
= -2i\lim_{R \to \infty} Re^{-\log(R + \sqrt{R^2 + \frac{1}{4\Omega^2}}) + \log k_1}
= -\frac{1}{2\Omega}.
\end{align*} 

Writing the expression (\ref{e2kappa}) for $e^{2\kappa}$ in terms of the axis-adapted basis $\{\tilde{a},\tilde{b}\}$, we get
\begin{align*}
e^{2\kappa}= K_0\frac{\tilde{\Theta}\left( \tilde{u}\right)\tilde{\Theta}\left( \tilde{u}+\int^{i \bar{z}}_{-i z}\tilde{\omega}\right)}{\tilde{\Theta}\left( 0\right)\tilde{\Theta}\left( \int^{i \bar{z}}_{-i z}\tilde{\omega}\right)}
e^{\frac{1}{2}\int_{\Gamma}^{\prime}\diff\kappa_1 h(\kappa_1)\int_{\Gamma}h(\kappa_2)
\big(\frac{\partial\tilde{\omega}_{\kappa^{+}_{1}\kappa^{-}_{1}}}{\partial\kappa_1}(\kappa_2^{+})-\frac{\diff\kappa_2}{(\kappa_2-\kappa_1)^2} \big)}.
\end{align*}
The expression (\ref{K0 formula}) for $K_0$ follows by letting $\rho\to 0$ and using that $e^{2\kappa}=1$ on the rotation axis. This completes the proof of the formulas for $a_0$ and $K_0$ and thus also of Theorem \ref{mainth1}.
\end{proof}

\section*{Acknowledgement}
L.P. would like to thank Mats Ehrnstr\"{o}m and Yuexun Wang  for their hospitality at NTNU where part of the research presented in this paper was conducted. 
Support is acknowledged from the European Research Council, Grant Agreement No.~682537 and the Swedish Research Council, Grant No.~2015-05430, and the G\"oran Gustafsson Foundation. 

\bibliographystyle{plain}
\bibliography{is}

\begin{thebibliography}{99}
\small

\bibitem{BW1969}
{\sc M.~Bardeen and R.~V. Wagoner}, {\em Uniformly rotating disks in general
  relativity}, Astrophys. J., 158 (1969), pp.~L65--L69.

\bibitem{BW1971}
\leavevmode\vrule height 2pt depth -1.6pt width 23pt, {\em Relativistic disks.
  I. Uniform rotation}, Astrophys. J., 167 (1971), pp.~359--423.

\bibitem{B2000}
{\sc J.~r. Bi\v{c}\'ak}, {\em Selected solutions of {E}instein's field
  equations: Their role in general relativity and astrophysics}, in Einstein's
  field equations and their physical implications, vol.~540 of Lecture Notes in
  Phys., Springer, Berlin, 2000, pp.~1--126.

\bibitem{E1968}
{\sc F.~J. Ernst}, {\em New formulation of the axially symmetric gravitational
  field problem}, Phys. Rev., 167 (1968), pp.~1175--1178.

\bibitem{farkas1992riemann}
{\sc H.~M. Farkas and I.~Kra}, {\em {Riemann Surfaces}}, Springer, 1992.

\bibitem{fay2006theta}
{\sc J.~D. Fay}, {\em {Theta Functions on Riemann Surfaces}}, vol.~352,
  Springer, 2006.

\bibitem{fokas1997unified}
{\sc A.~S. Fokas}, {\em A unified transform method for solving linear and
  certain nonlinear {PDEs}}, {\it Proc. Roy. Soc. Lond.} A 453 (1997), pp.~1411--1443.

\bibitem{klein1998physically}
{\sc C.~Klein and O.~Richter}, {\em Physically realistic solutions to the
  {E}rnst equation on hyperelliptic {R}iemann surfaces}, Phys. Rev. D, 58
  (1998), p.~124018.

\bibitem{klein2005ernst}
{\sc C.~Klein and O.~Richter}, {\em {Ernst Equation and Riemann Surfaces:
  Analytical and Numerical Methods}}, vol.~685, Springer Science \& Business
  Media, 2005.

\bibitem{korotkin2000theta}
{\sc D.~A. Korotkin and V.~B. Matveev}, {\em Theta function solutions of the
  Schlesinger system and the Ernst equation}, Funct. Anal. Appl., 34 (2000), pp.~252--264.

\bibitem{lenells2011boundary}
{\sc J.~Lenells}, {\em Boundary value problems for the stationary axisymmetric
  {E}instein equations: A disk rotating around a black hole}, Comm. Math. Phys., 304 (2011), pp.~585--635.

\bibitem{lenells2010boundary}
{\sc J.~Lenells and A.~S. Fokas}, {\em Boundary-value problems for the
  stationary axisymmetric {E}instein equations: a rotating disc}, Nonlinearity,
  24 (2011), pp.~177--206.

\bibitem{meinel2008relativistic}
{\sc R.~Meinel, M.~Ansorg, A.~Kleinw{\"a}chter, G.~Neugebauer, and D.~Petroff},
  {\em {Relativistic Figures of Equilibrium}}, Cambridge University Press
  Cambridge (UK), 2008.

\bibitem{N1996}
{\sc G.~Neugebauer}, {\em Gravitostatics and rotating bodies}, in General
  relativity ({A}berdeen, 1995), Scott. Univ. Summer School Phys., Edinburgh,
  1996, pp.~61--81.

\bibitem{neugebauer1993einsteinian}
{\sc G.~Neugebauer and R.~Meinel}, {\em The Einsteinian gravitational field of
  the rigidly rotating disk of dust}, The Astrophys. J., 414 (1993),
  pp.~L97--L99.

\bibitem{neugebauer1994general}
\leavevmode\vrule height 2pt depth -1.6pt width 23pt, {\em General relativistic
  gravitational field of a rigidly rotating disk of dust: Axis potential, disk
  metric, and surface mass density}, Phys. Rev. Lett., 73 (1994),
  p.~2166.

\bibitem{neugebauer1995general}
\leavevmode\vrule height 2pt depth -1.6pt width 23pt, {\em General relativistic
  gravitational field of a rigidly rotating disk of dust: Solution in terms of
  ultraelliptic functions}, Phys. Rev. Lett., 75 (1995), p.~3046.

\end{thebibliography}

\end{document}